\documentclass[twocolumn, showpacs,superscriptaddress, nofootinbib]{revtex4-2}
\usepackage{multirow}
\usepackage{mathrsfs}
\usepackage{amsmath}
\usepackage[colorlinks=true, allcolors=blue]{hyperref}
\usepackage{verbatim}
\usepackage{amsfonts}
\usepackage{amssymb}
\usepackage{amsthm}
\usepackage{bm}
\usepackage{dsfont}
\usepackage{xparse}
\usepackage{graphicx}
\usepackage{xcolor}
\usepackage{textcase}
\usepackage{url}
\usepackage{lipsum}
\usepackage{epstopdf}
\usepackage{footnote}
\usepackage{color}  

\newtheorem{theorem}{Theorem}
\newtheorem{lemma}[theorem]{Lemma}
\newtheorem{conjecture}[theorem]{Conjecture}
\newtheorem{definition}[theorem]{Definition}

\def\a{\alpha}
\def\b{\beta}
\def\g{\gamma}
\def\G{\Gamma}

\def\f{\phi}
\def\h{\eta}

\def\l{\lambda}
\def\L{\Lambda}
\def\m{\mu}
\def\r{\rho}
\def\s{\sigma}

\def\j{\varphi}
\def\q{\theta}
\def\x{\xi}
\def\y{\psi}
\def\z{\zeta}

\def\tr{\text{tr}}
\def\erf{\text{erf}}
\def\ox{\otimes}
\def\id{\mathds{1}}
\def\limn{\lim_{N \to \infty}}
\def\ton{\underset{N \to \infty}{\longrightarrow}}

\newcommand{\avg}[1]{\langle{#1}\rangle}
\newcommand{\ket}[1]{| {#1} \rangle} 
\newcommand{\bra}[1]{\langle {#1} |} 
\newcommand{\braket}[2]{\langle {#1} \vphantom{#2} | {#2} \vphantom{#1} \rangle}
\newcommand{\mel}[3]{\langle {#1} \vphantom{#2} | {#2} \vphantom{#3} | {#3} \rangle}
\newcommand{\bket}[1]{\big| {#1} \big\rangle}
\newcommand{\bbra}[1]{\big\langle {#1} \big|}
\DeclareMathOperator*{\Span}{Span}
\DeclareMathOperator{\sgn}{sign}

\usepackage{enumitem}
\usepackage{float}

\begin{document}
\title{Quantum theory at the macroscopic scale}
\author{Miguel Gallego}
\email{miguel.gallego.ballester@univie.ac.at}
\affiliation{University of Vienna, Faculty of Physics, Vienna Center for Quantum Science and Technology, Boltzmanngasse 5, 1090 Vienna, Austria}
\affiliation{University of Vienna, Vienna Doctoral School in Physics, Bolztmanngasse 5, 1090 Vienna, Austria}
\author{Borivoje Daki\'c}
\email{borivoje.dakic@univie.ac.at}
\affiliation{University of Vienna, Faculty of Physics, Vienna Center for Quantum Science and Technology, Boltzmanngasse 5, 1090 Vienna, Austria}
\affiliation{
Institute for Quantum Optics and Quantum Information (IQOQI),
Austrian Academy of Sciences, Boltzmanngasse 3,
A-1090 Vienna, Austria}

\date{\today}

\begin{abstract}
The quantum description of the microscopic world is incompatible with the classical description of the macroscopic world, both mathematically and conceptually. Nevertheless, it is generally accepted that classical mechanics emerges from quantum mechanics in the macroscopic limit. In this paper, we challenge this perspective and demonstrate that the behaviour of a macroscopic system can retain all aspects of the quantum formalism, in a way that is robust against decoherence, particle losses and coarse-grained (imprecise) measurements. This departure from the expected classical description of macroscopic systems is not merely mathematical but also conceptual, as we show by the explicit violation of a Bell inequality and a Leggett-Garg inequality.
\end{abstract}

\maketitle

\section{Introduction}
Quantum mechanics is one of the most successful scientific theories, and it is generally accepted as more fundamental than classical mechanics. However, quantum behaviour is not observed at larger scales, where classical physics provides a better description. To explain the macroscopic world we perceive in our everyday life, it is believed that there must exist a quantum-to-classical transition or, in other words, that classical mechanics must somehow emerge from quantum mechanics in the macroscopic limit (somewhat in the spirit of Bohr's correspondence principle \cite{bohr}). The questions of when and how exactly this transition occurs, in spite of being active for almost a hundred years now \cite{schrodinger}, are still debated today \cite{schlosshauer}.

\maketitle

One way to explain the emergence of classicality is to introduce genuine non-quantum effects (thus modifying quantum theory), such as the dynamical \cite{bassi} or gravitationally induced collapse \cite{penrose} of the wave function. Another standard way is to investigate the quantum-to-classical transition from within quantum theory (which is also the subject here). One of the most famous approaches in this direction is the decoherence mechanism \cite{zurek, zurek2003, schlosshauerbook}, which shows that macroscopic systems, being hard to isolate, lose coherence in their interaction with the environment. Consequently their description becomes effectively classical. Complementary to this approach is the coarse-graining mechanism \cite{kb, kb2010}, which shows that outcomes of macroscopic measurements admit a classical description when their resolution is limited even for perfectly isolated systems. While the main focus of the decoherence mechanism is on the dynamics, the coarse-graining approach focuses on the kinematic aspect of the transition to classicality. One way or another, the common conclusion is that quantum effects disappear in the macroscopic limit. Here, we want to challenge this view and ask \emph{how much decoherence or coarse-graining is needed to observe classicality?} Typically, one takes a ``transition'' parameter, such as the size of the system: for example, the number of microscopic constituents $N$ of a large system. Then, the standard statement is rather qualitative, positing that for a measurement resolution much greater than $\sqrt N$, one obtains an effective classical description \cite{koflerphd}. Nevertheless, the precise mathematical meaning of what ``much greater than'' means is vague in a concrete experimental situation. Our goal here is to make such statements mathematically precise and to show a well-defined macroscopic scale at which large quantum systems can fully preserve a quantum description, in the sense of the typical ingredients of quantum theory such as the notion of \emph{Hilbert space}, the \emph{Born rule} and the \emph{superposition principle}. Furthermore, we will show that these are not merely mathematical artefacts but genuine quantum phenomena by explicitly showing the violation of \emph{Bell} \cite{bell} and \emph{Leggett-Garg} \cite{leggettgarg} inequalities for such systems.

\section{Coarse-grained measurements and the quantum-to-classical transition}
The concept of coarse-grained measurement in quantum mechanics appeared already long ago as an attempt to address Born's rule using the relative frequency operator (fuzzy or coarse-grained observables \cite{hartle}). Continuing this line of inquiry, many subsequent works have followed, including discussions related to the weak and strong laws of large numbers \cite{hartle, farhigoldstone}, as well as the quantum-to-classical correspondence (see \cite{zurek, zurek2003} and references therein). A significant breakthrough in the operational understanding of the emergence of classicality through the coarse-graining mechanism came due to Kofler and Brukner \cite{kb, kb2010}. In their works, the authors focus on classicality via the notion of macroscopic realism \cite{leggettgarg} or Bell's local realism \cite{bell}. Such a framework opens an operational route to study the observability of genuine quantum effects, signalled by the violation of Leggett-Garg \cite{leggettgarg} or Bell \cite{bell} inequalities. Their main result is that with a sufficient level of coarse-graining (much greater than $\sqrt{N}$) the outcomes of the experiment can be modelled by a classical distribution (at least for finite dimensional systems) \cite{kb, kb2010}. In particular, outcomes of successive measurements on a single system satisfy macroscopic realism (i.e. satisfy all Leggett-Garg inequalities) \cite{kb}, and local measurements on a bipartite system satisfy local realism (i.e. satisfy all Bell inequalities) \cite{kb2010}. Similar results can be found in the context of more general (post-quantum) theories in \cite{ramanathan, barbosa}. On the other hand, if the level of coarse-graining is just right, namely of the order $\sqrt N$, then we have the following facts: for \emph{independent and identically distributed} (IID) pairs of (finite-dimensional) quantum systems, the coarse-grained quantum correlations satisfy Bell locality \cite{nw} (an analogous result can be shown for quantum contextuality \cite{hensonsainz}), while non-IID quantum states can exhibit nonlocal correlations as shown in our earlier works \cite{gd, gallego}. Furthermore, these works show that an entire family of (quantum) non-central limit theorems arises in such non-IID scenarios, raising an interesting situation in which, although being coarse-grained, macroscopic quantum systems exhibit quantum phenomena. We formalized this via the idea of \emph{macroscopic quantum behaviour} \cite{gd,gallego}, a property of a system in the macroscopic limit that retains the mathematical structure of quantum theory under the action of decoherence, particle losses, and coarse-graining. We have shown that it is possible to preserve the typical ingredients of the quantum formalism in the macroscopic limit, such as the Born rule, the superposition principle, and the incompatibility of the measurements. 

In this paper, we build on our previous work and develop a \emph{unified framework for quantum theory at the macroscopic scale}, which describes the theory of successive concatenation of measurements in the macroscopic limit. We use the formalism of Kraus operators in the limit Hilbert space, which allows us to show a violation of a Leggett-Garg inequality in the macroscopic limit, which in turn shows the incompatibility of coarse-grained measurements (at the $\sqrt N$ level) with a macrorealistic description of the large quantum system.

\subsubsection*{Decoherence vs. coarse-graining}
Before proceeding further, we make some remarks on decoherence and its relation to coarse-graining. The decoherence mechanism stresses the role of dynamical loss of quantum coherence due to the (instantaneous) interaction with the environment followed by the einselection process of the pointer basis \cite{zurek, zurek2003}. On the other hand, the coarse-graining mechanism focuses on a kinematical aspect of the transition to classicality by measuring collective coarse-grained observables.  Nevertheless, if decoherence is understood broadly as a mechanism of ``classicalization" due to interaction with the environment, then coarse-graining can be seen as an instance of such a mechanism. Namely, suppose such a process is described by an interaction Hamiltonian of the type $H=H_S\otimes H_E$, with $H_S=\sum_i h_i$ describing the collective Hamiltonian of the large quantum system (here $h_i$ refers to the operators associated with local, microscopic constituents). In that case, the decoherence model effectively reduces to the measurement model of collective coarse-grained observables. This will be precisely our model of measurement, which we will introduce in the next sections. Therefore, essential aspects of the standard decoherence mechanism are incorporated in our study through coarse-graining of the measurements, and this is the standard argument to draw the parallel between the two approaches to the quantum-to-classical transition \cite{caslav}. Notice that the system interacts collectively with the environment in that case. On the other hand, microscopic constituents can as well be independently subjected to a decoherence channel (such as the dephasing or depolarizing channel) \cite{frowisstable}. We will also include this mechanism in our study and refer to it as \emph{local decoherence} to distinguish it from the standard decoherence mechanism (a precise definition will be provided to avoid misunderstandings). Our aim is to show robustness of quantum phenomena in the macroscopic limit against both mechanisms.

The paper is organized as follows. In Section \ref{sec:setup}, we define the setting under consideration and specify all relevant assumptions. The notion of \emph{macroscopic quantum behaviour} is introduced in Section \ref{sec:mqb}, encompassing the quantum properties of systems in the macroscopic limit, with concrete examples provided. In Section \ref{sec:deviceindependent}, we further analyse the properties of the system, demonstrating the genuineness of macroscopic quantum behaviour through device-independent tests such as the explicit violations of Bell inequalities as well as Leggett-Garg inequalities. Finally, we conclude with final remarks and open questions in Section \ref{sec:outlook}.

\section{Setup}
\label{sec:setup}
We consider a macroscopic quantum measurement scenario analogous to the one presented in \cite{nw, gd, gallego} (see Figure \ref{fig:meas}). The setting consists of two parts: a quantum system $\mathsf{S}$ and a quantum measurement apparatus $\mathsf{M}$. To model a realistic situation in the absence of perfect control, we assume these satisfy certain assumptions. 

\subsubsection*{Macroscopic system}
First, we assume the system $\mathsf{S}$ satisfies the following conditions:
\begin{enumerate}[label=(\textit{\roman*})]
    \item \emph{Large $N$.} The system is composed of a large number $N$ of identical particles or subsystems with associated Hilbert space $\mathfrak{h}$. We describe the state of the system with a density matrix $\rho_N \in \mathcal{D}(\mathfrak{h}^{\otimes N})$, i.e. a positive, self-adjoint bounded linear operator on $\mathfrak{h}^{\otimes N}$ satisfying $\tr \, \rho_N =1$.
    \label{ass:largen}

    \item \emph{Local decoherence.} The system is subject to independent, single-particle decoherence channels $\Gamma$, such as the depolarizing or the dephasing channel. The effective state thus becomes $\Gamma^{\otimes N} ( \rho_N )$.
    \label{ass:decoh}

    \item \emph{Particle losses.} Each individual particle has a probability $p \in (0,1]$ of reaching the measurement apparatus, while $1-p$ is the probability of being lost (we assume $p > 0$ to avoid the trivial case where no particles reach the apparatus). Therefore, in each run of the experiment, only a number $M \leq N$ of particles reach the measurement apparatus with a probability $f_N(M) = {N \choose M} p^{M} (1-p)^{N-M}$. Consequently, the state received by the measurement apparatus is of the form $\tr_{1 \dots N-M} \r_N$ (where $\tr_i$ is understood as the partial trace over the Hilbert space $\mathfrak{h}$ of the $i$-th particle).
    \label{ass:loss}
\end{enumerate}
Given these assumptions, for an initial state of the system $\r_N \in \mathcal{D}(\mathfrak{h}^{\ox N})$, the effective state of the system can be written in Fock-like space $\mathcal{F}_N = \oplus_{M=0}^N \mathfrak{h}^{\ox M}$ as  
\begin{align}
    \L_N(\r_N) = \bigoplus_{M=0}^N  f_N(M) \,  \sum_{\pi \in \mathfrak{S}_N} \frac{\tr_{\pi(1) \dots \pi(N-M)} \Big[  \G^{\otimes N} ( \rho_N )  \Big] }{N!}  \, ,
\label{eq:LambdaN}
\end{align}
where $\mathfrak{S}_N$ is the symmetric (permutation) group of $N$ elements. 

\subsubsection*{Coarse-grained measurements}
To model the macroscopic measurements, we assume the measurement apparatus $\mathsf{M}$ satisfies the following conditions:
\begin{enumerate}[label=(\textit{\roman*})]
    \setcounter{enumi}{3}
    \item \emph{Collective measurement}. The measurement setting of the apparatus is given by a single-particle observable, i.e. a Hermitian bounded operator $A \in \mathcal{B}(\mathfrak{h})$. We denote by $A \ket{a}_A = a \ket{a}_A$ the diagonalization of $A$, where $a \in \mathbb{R}$ are its eigenvalues and $\ket{a}_A \in \mathfrak{h}$ are its eigenstates. We denote by $\mathcal{A} \subseteq \mathcal{B}(\mathfrak{h})$ the set of (experimentally) accessible single-particle observables.
    \label{ass:coll}
    
    \item \emph{Intensity measurement.} Given a measurement setting $A \in \mathcal{A}$, the measurement apparatus measures the intensity $\sum_{i} a_i$, i.e. the sum of individual outcomes. The corresponding observable in Fock-like space is therefore the intensity observable $I_N(A) = \oplus_{M=1}^N \sum_{i=1}^M A_i$, where $A_i= \mathbb{I} \otimes\dots  \otimes A\otimes  \dots \otimes \mathbb{I}$ is the operator that acts with $A$ on the $i$-th particle and with the identity on the rest. 
    \label{ass:intensity}
    
    \item \emph{Coarse-graining.} The measuring scale for the intensity $\sum_{i} a_i$ has a limited resolution of the order of $\sqrt{N}$ (the square-root of the total number of particles), meaning that it cannot distinguish between values that differ by approximately less than $\sqrt{N}$. 
    \label{ass:cg}
\end{enumerate}
\begin{figure}
        \centering
        \includegraphics[width=0.45\textwidth]{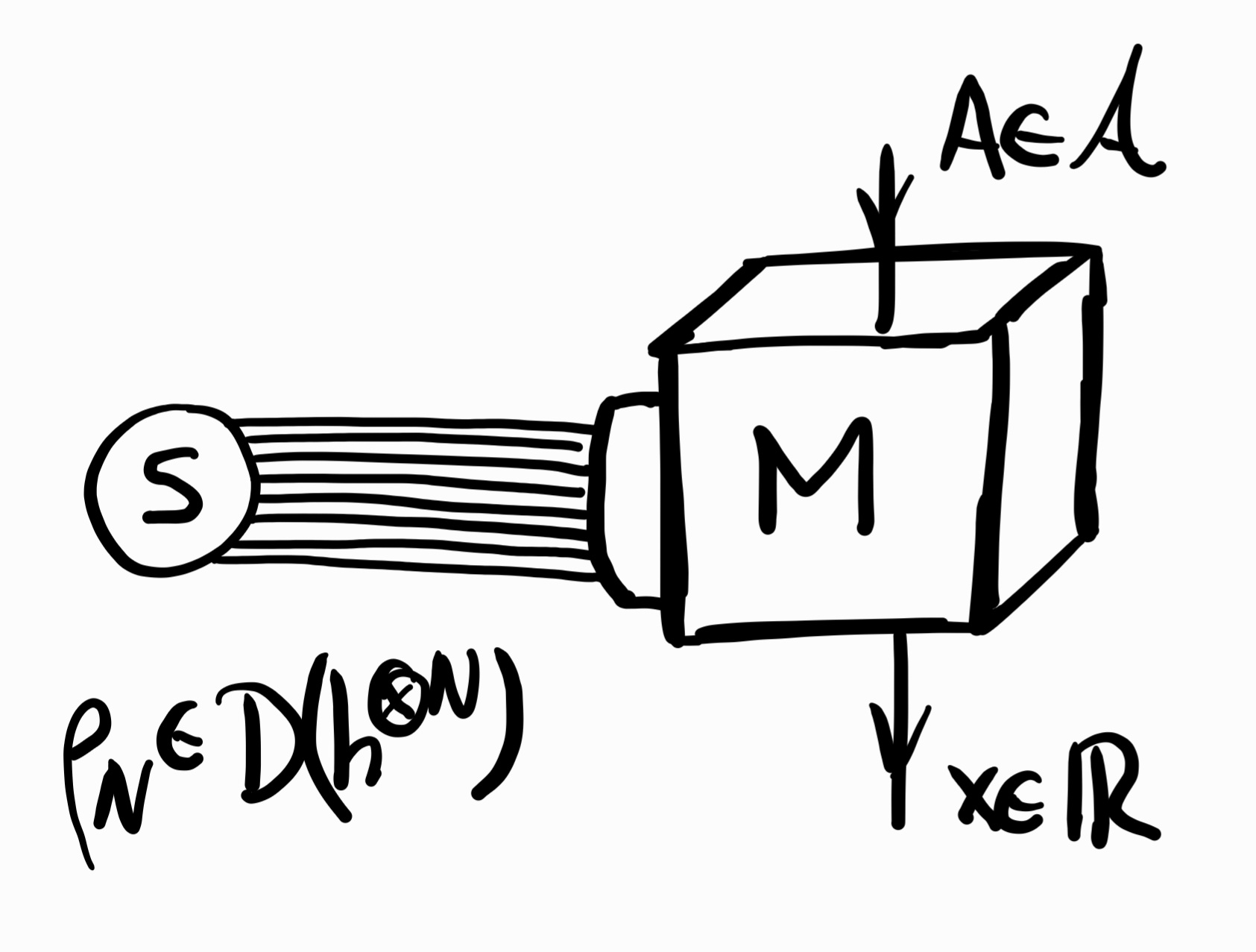}
        \caption{\textbf{Macroscopic measurement.} A macroscopic quantum system $\mathsf{S}$ is sent into a macroscopic quantum measurement apparatus $\mathsf{M}$, which implements a coarse-grained measurement.}
        \label{fig:meas}
\end{figure}
\subsubsection*{Measurement model}
In order to implement these assumptions explicitly, we follow von Neumann's model of quantum measurements \cite{vn}. First, we assume that the measurement apparatus $\mathsf{M}$ couples the system $\mathsf{S}$ to an auxiliary system $\mathsf{P}$ called the pointer, initially in a state $\ket{\Phi} \in L^2(\mathbb{R})$ centered around zero in position basis and with a standard deviation of order $\sqrt{N}$. For simplicity, we take this state to be a Gaussian with standard deviation $\s \sqrt{N}$ for some $\s > 0$, i.e.
\begin{align}
    \Phi_N(x) = \frac{1}{(2 \pi N \s^2)^{1/4}} e^{-x^2/(4 N  \s^2)} \, .
\label{eq:gaussian}
\end{align}
The coupling between the system and the pointer is described by the Hamiltonian $H(t) = \g(t) \oplus_{M=0}^N \sum_{i=1}^M A_i \ox P$, where $\g(t)$ is a nonzero function only for a short time satisfying $\int dt \, \g(t) =1$ and $P$ is the momentum operator of the pointer. After unitary interaction, the position of the pointer is translated to a distance equal to the value of the system's intensity $\sum_{i} a_i$ (i.e., an eigenvalue of the intensity operator $I_N(A)$). Finally, the pointer's position is measured, obtaining a value $x \in \mathbb{R}$. For a single-particle observable $A \in \mathcal{A}$ (which represents a measurement setting) and an initial state of the system $\r_N \in \mathcal{D}(\mathfrak{h}^{\ox N})$, we set $X_N$ to be the random variable associated to the measured value $x$ (the measurement result on $\mathsf{P}$). We shall use the symbol ``$\sim$" to denote ``distributed according to", and as shown in the Appendix A of the supplementary material, we have $X_N \sim P(x|A)$ where
\begin{align}
    P_N(x |A) = \tr \Big[ K_N(x| A) \, \L(\r_N) \,  K_N^\dagger(x| A) \Big].
\label{eq:distributionx}
\end{align}
Here, the Kraus operators $K_N(x|A)$ are given by
\begin{align}
    \bigoplus_{M=0}^N \sum_{a_1 \dots a_M} \Phi_N \left( x - \sum_{i=1}^M a_i \right) \, \Pi_{a_1 | A}  \otimes \dots \otimes \Pi_{a_M | A}\, ,
\label{eq:kraus}
\end{align}
with $\Pi_{a|A}=\ket{a}_A \, _A\bra{a}$ being the eigenprojectors of $A$. These Kraus operators define a \emph{positive operator-valued measure} (POVM) with elements $E_N (x | A) = K_N^\dagger (x|A) K_N (x|A) $, normalized so that $\int dx \, E_N(x|A) = \mathbb{I}$. The corresponding (normalized) post-measurement state of the system is given by the standard expression
\begin{align}
    \frac{K_N (x | A ) \, \L(\r_N) \, K_N^\dagger(x | A ) }{\tr \big[ K_N(x | A) \, \L(\r_N) \, K_N^\dagger(x | A) \big] }\, ,
    \label{eq:rhon(x,s)}
\end{align}
To summarize, the effective state $\L(\r_N)$ describes the system $\mathsf{S}$ under assumptions \ref{ass:largen}, \ref{ass:decoh} and \ref{ass:loss}, while the Kraus operators $K_N(x|A)$ describe the measurement apparatus $\mathsf{M}$ under assumptions 
\ref{ass:coll}, \ref{ass:intensity} and \ref{ass:cg}.

\subsubsection*{Macroscopic limit}
The macroscopic limit corresponds to the limit of an infinite number of particles, i.e., $N\rightarrow\infty$. Nevertheless, before proceeding further, let us consider possible scenarios in such a limit. The random variable $X_N$ will generally not converge in distribution as $N \to \infty$. To illustrate this, let the initial state of the system be an \emph{independent and identically distributed} (IID) state $\r_N = \r^{\ox N}$ for some $\r \in \mathcal{D}(\mathfrak{h})$. In this case, the distribution of the intensity $\sum_{i=1} a_i$ does not converge in general, unless one subtracts to it the mean value $\langle \sum_{i=1} a_i \rangle$ and divides by $\sqrt{N}$, just like in the central limit theorem \cite{billingsley}. Therefore, to ensure convergence, we take an affine transformation of $X_N$, namely we consider a family of random variables of the form $ \l_N X_N + \m_N$, and choose the parameters $\l_N, \m_N \in \mathbb{R}$ (independent of the measurement setting and initial state of the system) in a way such that $\l_N X_N + \m_N$ converges in distribution. Once this is fixed, the corresponding probability density function and the Kraus operators given in Equations \eqref{eq:distributionx} and \eqref{eq:kraus} transform as
\begin{align*}
    P_N(x|A) &  \mapsto \l_N^{-1}  \,  P_N \big( \l_N^{-1}(x-\m_N) |A \big) \, ,
\end{align*}
and
\begin{align*}
    K_N(x|A) & \mapsto \l_N^{-1/2} \, K_N \big( \l_N^{-1}(x-\m_N) | A \big) \, ,
\end{align*}
respectively.
\begin{figure}
        \centering
        \includegraphics[width=0.45\textwidth]{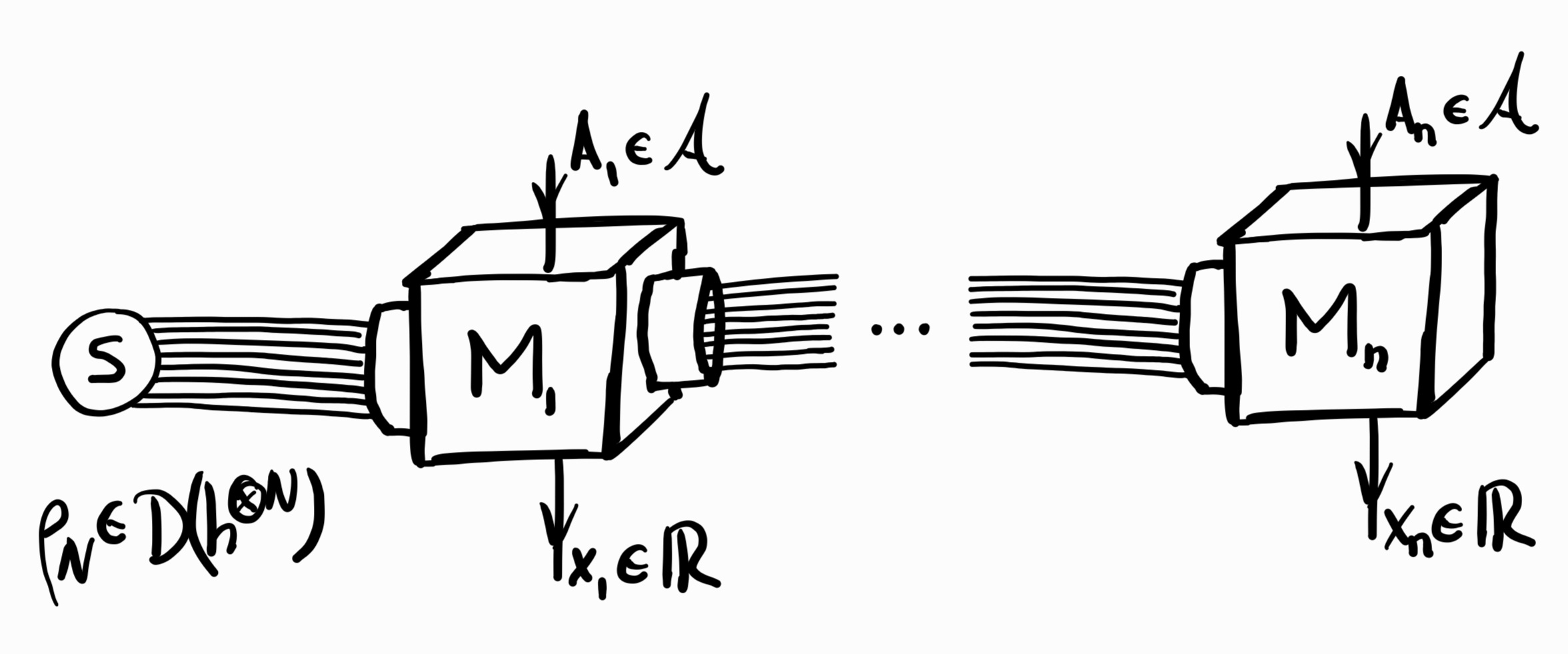}
        \caption{\textbf{$n$ consecutive macroscopic measurements.} A macroscopic quantum system $\mathsf{S}$ is successively sent into measurement apparatuses $\mathsf{M}_1, \dots, \mathsf{M}_n$ which implement corresponding coarse-grained measurements.}
        \label{fig:nmeas}
\end{figure}

\subsubsection*{$n$ consecutive measurements}

Finally, we are ready to present the most general scenario where a number $n$ of successive measurements are performed (see Figure \ref{fig:nmeas}). In this case, the system $\mathsf{S}$, subjected to assumptions \ref{ass:largen} - \ref{ass:loss} as before, goes through measurement apparatuses $\mathsf{M}_1$, ..., $\mathsf{M}_n$, each of which satisfies assumptions \ref{ass:coll} - \ref{ass:cg}. For an initial state of the system $\r_N \in \mathcal{D}(\mathfrak{h}^{\ox N})$ and a sequence of $n$ single-particle observables $(A_1, \dots , A_n)$ ($n$ measurement settings) let $\vec{X}_N = (X_{N}^{(1)}, \dots, X_{N}^{(n)})$ be the random vector associated to the measurement outcomes $(x_1, \dots, x_n)$. Then $\vec{X}_N \sim P(x_1, \dots, x_n | A_1, \dots, A_n)$, where the distribution $P(x_1, \dots, x_n | A_1, \dots, A_n)$ is given by
\begin{align}
     \tr  \bigg[ K_N^{(n)} \L^{(n)} \Big( \dots K_N^{(1)} \, \L^{(1)}(\r_N) \,  K_N^{(1) \dagger} \dots \Big)   K_N^{(n) \dagger} \bigg] \, .
\label{eq:distributionvec}
\end{align}
Here $K_N^{(j)} = K_N(x_j |A_j)$ and $\L^{(j)}$ is the map \eqref{eq:LambdaN} extended to Fock-like space (defined by loss probability $p_j$ and decoherence channel $\G_j$), i.e.,
\begin{align*}
    \L^{(j)} \left( \bigoplus_{M=0}^N \, \r_M  \right) = \bigoplus_{M=0}^N \L_M^{(j)} ( \r_M) \, ,
\end{align*}
where
\begin{align}
    \L_M^{(j)} (\r_M) = \bigoplus_{J=0}^M f_M^{(j)}(J) \sum_{\pi \in \mathfrak{S}_M} \frac{\tr_{\pi(1) \dots \pi(M-J)} [ \G_j^{\ox M} (\r_M)]}{M!} 
\label{eq:Lambda}
\end{align}
and $f_M^{(j)}(J) = {M \choose J}p_j^J (1-p_j)^{M-J}$. As before, the random vector $\vec{X}_N$ may not converge in distribution; thus, we shall consider an affine transformation of the form $\vec{\l}_N \odot \vec{X}_N + \vec{\m}_N$, where $\odot$ denotes the Kronecker product (e.g. $(a,b) \odot (c,d) = (ac, bd)$), and choose the vectors $\vec{\l}_N$ and $\vec{\m}_N$ to facilitate convergence.

\section{Quantum theory at the macroscopic scale}
\label{sec:mqb}
We now argue that, in the context of the above scenario, it is possible to define a joint notion of convergence for states and measurements that \emph{preserves the complete mathematical structure of quantum theory in the macroscopic limit}. Furthermore, we will show how this formalism specifically applies to device-independent quantities such as \emph{correlations} (both spatial, as in Bell experiments, and temporal, as in Leggett-Garg experiments). In other words, we will show that, in the limit $N \to \infty$, states $\r_N \in \mathcal{D}(\mathfrak{h}^{\ox N})$ can be mapped to states $\rho \in \mathcal{D}(\mathfrak{H})$ in some ``limit" Hilbert space $\mathfrak{H}$ and Kraus operators $K_N(x|A)$ can be mapped to Kraus operators $K(x|A)$ acting on $\mathfrak{H}$ in a way such that the essential ingredients of quantum mechanics, including the Born rule, the superposition principle and the incompatibility of measurements, are retained. To do this, let us introduce the concepts of \emph{macroscopic quantum representation} and \emph{(robust) macroscopic quantum behaviour of order} $n$ (MQB$_n$).

\begin{definition}
    Given a closed subspace $\mathfrak{H}_N \subseteq \mathfrak{h}^{\otimes N}$ and a subset $\mathcal{A} \subseteq \mathcal{B}( \mathfrak{h} )$, a \emph{macroscopic quantum representation} is a limit of the form
    \begin{align*}
        \Big( \mathfrak{H}_N, \r_N, K_N(x|A) \Big) \ton  \Big( \mathfrak{H}, \r, K(x|A) \Big)
    \end{align*}
    for all $\r_N \in \mathcal{D}(\mathfrak{H}_N)$ and for all $A \in \mathcal{A}$, where
    \begin{itemize}
        \item  $\mathfrak{H}$ is a Hilbert space;
        \item the limit state $\r  \in \mathcal{D}$ is a ``linear" function of $\r_N$, in the sense that if the pure states $\ket{\y}_N$ and $\ket{\f}_N$ are mapped, respectively, to $\ket{\y}$ and $\ket{\f}$, then the linear combination $\a \ket{\y}_N + \b \ket{\f}_N$ is mapped to the linear combination $\a \ket{\y} + \b \ket{\f}$;
        \item the limit Kraus operators $K( x | A)  \in \mathcal{B}(\mathfrak{H})$ form a non-compatible set of measurements. 
    \end{itemize}
\end{definition}

\begin{definition}
    A closed subspace $\mathfrak{H}_N \subseteq \mathfrak{h}^{\otimes N}$ and a subset $\mathcal{A} \subseteq \mathcal{B}( \mathfrak{h} )$ possess \emph{MQB$_n$} for some $n \in \mathbb{N}$ if there exists a macroscopic quantum representation
    \begin{align*}
        \Big( \mathfrak{H}_N, \r_N, K_N(x|A) \Big) \to  \Big( \mathfrak{H}, \r, K(x|A) \Big)
    \end{align*}
    such that for every state $\rho_N \in \mathcal{D} (\mathfrak{H}_N)$ and for every sequence of measurement settings $(A_1, \dots, A_n ) \in \mathcal{A}^{\ox n}$, the random vector 
    \begin{align*}
        \vec{X}_N \sim \tr  \bigg[ K_N^{(n)}  \dots K_N^{(1)} \, \r_N \,  K_N^{(1) \dagger} \dots K_N^{(n) \dagger} \bigg] 
     \end{align*}
     (or an affine transformation thereof, i.e. $\vec{\l}_N \odot \vec{X}_N + \vec{\m}_N$ for some suitably chosen $\vec{\l}_N, \vec{\m}_N \in \mathbb{R}^n$), where $K_N^{(j)} = K_N(x_j |A_j)$, converges in distribution as $N \to \infty$ to some random vector
    \begin{align*}
         \vec{X} \sim \tr \big[  K^{(n)} \dots K^{(1)} \,  \rho \,  K^{(1) \dagger} \dots K^{(n) \dagger} \big] \, ,
    \end{align*}  
    where $K^{(j)} = K(x_j | A_j)$ (as given by the macroscopic quantum representation).
\end{definition}

In references \cite{gd, gallego}, only the case of a single measurement was considered ($n=1$), and the definition of ``MQB" given there corresponds to MQB$_1$ as defined here. Now we define a stronger notion, where the macroscopic quantum representation is robust against decoherence and losses as introduced in assumptions \ref{ass:decoh} and \ref{ass:loss} respectively.

\begin{definition}
    A closed subspace $\mathfrak{H}_N \subseteq \mathfrak{h}^{\otimes N}$ and a subset $\mathcal{A} \subseteq \mathcal{B}( \mathfrak{h} )$ possess {\bf \emph{robust}} MQB$_n$ for some $n \in \mathbb{N}$ if there exist a macroscopic quantum representation
    \begin{align*}
        \Big( \mathfrak{H}_N, \r_N, K_N(x|A) \Big) \to  \Big( \mathfrak{H}, \r, K(x|A) \Big)
    \end{align*}
    and $\epsilon > 0$ such that for every state $\rho_N \in \mathcal{D} (\mathfrak{H}_N)$, for every sequence of measurement settings $(A_1, \dots, A_n ) \in \mathcal{A}^{\ox n}$ and for every sequence of channels $(\L^{(1)}, \dots, \L^{(n)} )$ of the form \eqref{eq:Lambda} with $ \max_j \| \L^{(j)} - \rm{Id} \| \leq \epsilon$, the random vector 
    \begin{align*}
        \vec{X}_N \sim \tr  \bigg[ K_N^{(n)} \L^{(n)} \Big( \dots K_N^{(1)} \, \L^{(1)}(\r_N) \,  K_N^{(1) \dagger} \dots \Big) K_N^{(n) \dagger} \bigg] \, ,
     \end{align*}
     (or an affine transformation $\vec{\l}_N \odot \vec{X}_N + \vec{\m}_N$ for some suitably chosen $\vec{\l}_N, \vec{\m}_N \in \mathbb{R}^n$), where $K_N^{(j)} = K_N(x_j |A_j)$, converges in distribution as $N \to \infty$ to some random vector
    \begin{align*}
         \vec{X} \sim \tr \big[  K^{(n)} \dots K^{(1)} \,  \rho \,  K^{(1) \dagger} \dots K^{(n) \dagger} \big] \, ,
    \end{align*}  
    where $K^{(j)} = K(x_j | A_j)$ (as given by the macroscopic quantum representation).
\end{definition}

\subsubsection*{An example}
To illustrate these ideas, consider the case where $\mathfrak{h} = \mathbb{C}^2$ (thus, particles or subsystems are qubits). In this case, the single-particle projective measurements have two possible outcomes, which we label $+1$ and $-1$. Let us define the $N$-particle Dicke states \cite{dicke}
\begin{align*}
\ket{N,k} := \frac{1}{\sqrt{{N \choose k}}} \Big( \ket{ \underbrace{1\dots1}_{k} 0\dots0 } + \text{permutations} \Big) \, ,
\end{align*}
the subspace of $(\mathbb{C}^2)^{\ox N}$ generated by the first $d$ Dicke states 
\begin{align*}
    \mathfrak{D}_N := \Span \{ \ket{N,k} \, , \, k = 0 , 1 , \dots, d-1  \} 
\end{align*}
and the set of non-diagonal observables on $\mathbb{C}^2$
\begin{align*}
    \mathcal{ND} := \{ A \in \mathcal{B}(\mathbb{C}^2) : A^\dagger = A \quad \text{and} \quad \langle 0 | A | 1 \rangle \neq 0 \} \, .
\end{align*}
Then we have the following result:
\begin{theorem}
    The Dicke subspace $\mathfrak{D}_N$ with dimension $d \ll N$ (in the sense that $\lim_{N \to \infty} d/ N = 0$, which holds, for instance, if $d$ is fixed) and the set $\mathcal{ND}$ possess robust MQB$_1$.
\label{thm:mqb1}
\end{theorem}
\begin{proof}
    See Appendix B in the supplementary material.
\end{proof}
In particular, as shown in the Appendix B of the supplementary material, the MQB$_1$ is given by the macroscopic quantum representation
    \begin{align*}
        \Big( \mathfrak{D}_N, \ket{N,k}, K_N(x|A) \Big) \to  \Big( L^2(\mathbb{R}), \ket{k}, K(x|A) \Big) \, .
    \end{align*}
    Here, $L^2(\mathbb{R})$ is the space of square-integrable functions, $\ket{k}$ are number states (energy eigenstates of the quantum harmonic oscillator) and
    \begin{align}
        K(x|A) =  \frac{e^{-(X_\j-x)^2/(2 \b^2)}}{(\pi \b^2)^{1/4}} \, ,
    \label{eq:limitkraus}
    \end{align}
    where $X_\j = X \cos \j + P \sin \j$ are phase space observable in terms of position $X$ and momentum $P$ observables and angle $\j = \arg \langle 0 | A | 1 \rangle$ and
    \begin{align}
        \b^2 = \frac{\s^2 + p \mel{0}{\G^\dagger(A^2)}{0} - p^2 \mel{0}{\G^\dagger(A)}{0}^2}{p^2 |\mel{0}{\Gamma^\dagger(A)}{1} |^2} - 1 \, ,
    \label{eq:beta}
    \end{align}
    in terms of the probability $p$ and decoherence channel $\G$ defined in Eq. \eqref{eq:LambdaN}. We conjecture that this macroscopic quantum representation  for the considered system constitutes a robust MQB$_n$ for all $n$:
\begin{conjecture}
    The Dicke subspace $\mathfrak{D}_N$ (with dimension $d \ll N$) and the set $\mathcal{ND}$ possess robust MQB$_n$ for all $n \in \mathbb{N}$.
\label{conj}
\end{conjecture}
We prove a weaker result that, together with Theorem \ref{thm:mqb1}, supports Conjecture \ref{conj}:
\begin{theorem}
    The Dicke subspace $\mathfrak{D}_N$ (with dimension $d \ll N$) and the set $\mathcal{ND}$ possess MQB$_n$ for all $n \in \mathbb{N}$.
\label{thm:mqbn}
\end{theorem}
\begin{proof}
    See Appendix C in the supplementary material.
\end{proof}
Moreover, as shown in Appendix C of the supplementary material, all MQB$_n$ are given by the same macroscopic quantum representation as before
\begin{align*}
    \Big( \mathfrak{D}_N, \ket{N,k}, K_N(x|A) \Big) \to  \Big( L^2(\mathbb{R}), \ket{k}, K(x|A) \Big)
\end{align*}
with $\b = \s$ (i.e., the value of $\b$ is given by Eq. \eqref{eq:beta} with $p=1$ and $\G = \text{Id}).$

\section{Tests of macroscopic non-classical behaviours}
\label{sec:deviceindependent}
We have shown that the above system, consisting of Dicke states $\mathfrak{D}_N$ and collective measurements (as long as they are not represented by diagonal hermitian operators), possesses  MQB$_n$ for all $n \in \mathbb{N}$ as well as \emph{robust} MQB$_1$, which lead us to conjecture that it does indeed possess robust MQB$_n$ for all $n \in \mathbb{N}$. These results strongly hint that the quantum nature of the system can be observed at the macroscopic scale. In order to make this statement concrete, we use tests that witnesses the non-classical nature of our system at the macroscopic scale. In particular, we will see that our system possesses nonlocal correlations (in the sense of Bell \cite{bell}), ruling out a local realistic description of the outcome statistics in the macroscopic limit. We also show a violation of a Leggett-Garg inequality \cite{leggettgarg}, ruling out a macroscopic realistic description of the statistics.

\begin{figure}[t]
        \centering 
        \includegraphics[width=0.45\textwidth]{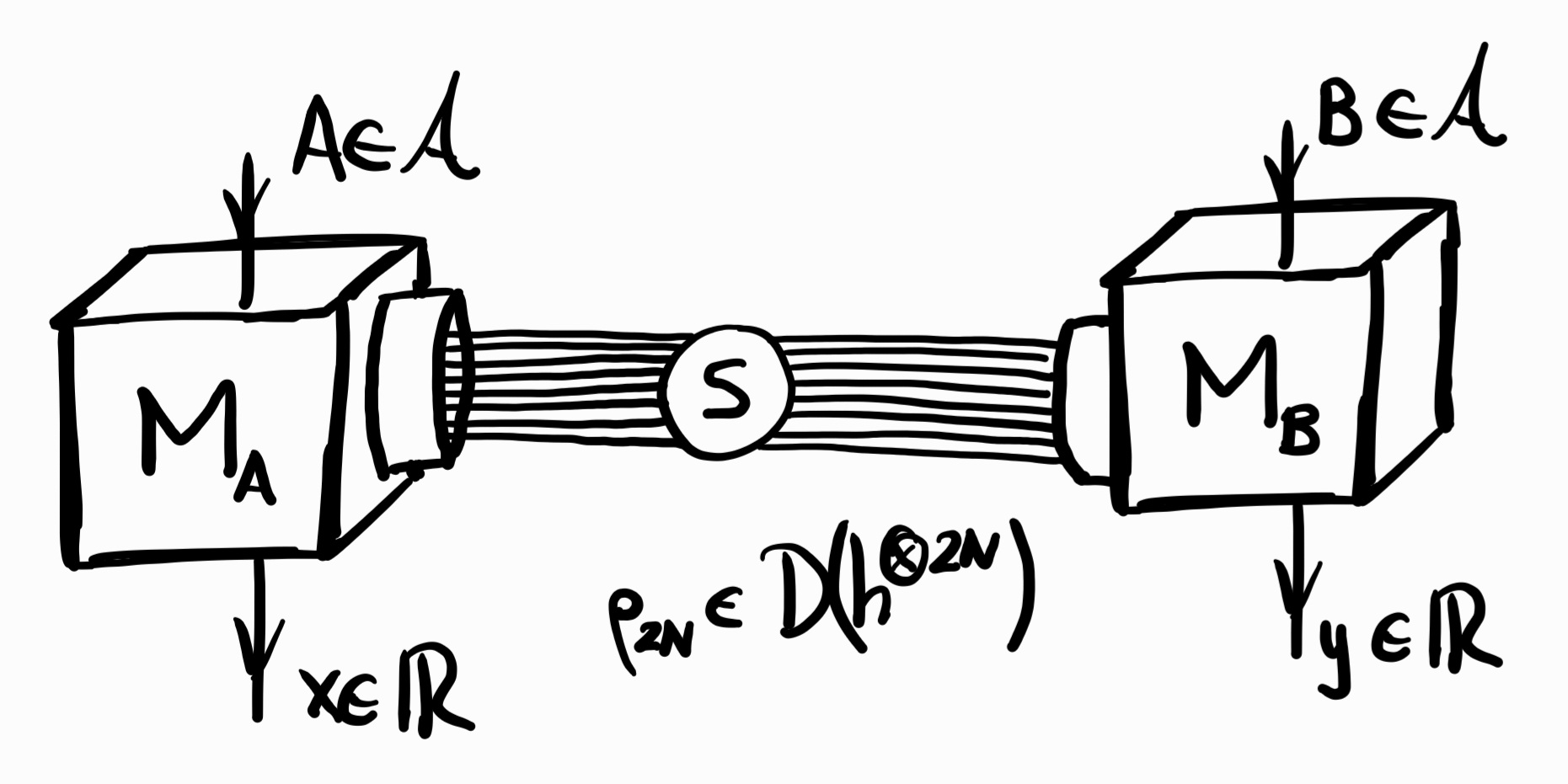} 
        \caption{\textbf{Macroscopic bipartite Bell experiment.} A macroscopic quantum system $\mathsf{S}$ is divided in two parts which are sent to measurement apparatuses $\mathsf{M}_A$ and $\mathsf{M}_B$ respectively, which implement coarse-grained measurements.}
        \label{fig:bell}
\end{figure}

\subsubsection*{Violation of a Bell inequality}
Consider the bipartite macroscopic measurement scenario depicted in Figure \ref{fig:bell}: a system $\mathsf{S}$, subject to assumptions \ref{ass:largen} - \ref{ass:loss}, is divided in two parts which are sent to measurement apparatuses $\mathsf{M}_A$ and $\mathsf{M}_B$, each of which satisfies assumptions \ref{ass:coll} - \ref{ass:cg}. Suppose the system is in a state of the form $\sum_{k =0}^{d-1} c_k \ket{N,k}_A \ox \ket{N,k}_B \in \mathfrak{D}_N \ox \mathfrak{D}_N$, and suppose that Alice selects a single-particle observable $A \in \mathcal{ND}$ obtaining an outcome $x \in \mathbb{R}$. Likewise, Bob selects $B \in \mathcal{ND}$ obtaining $y \in \mathbb{R}$. Then, by applying Theorem \ref{thm:mqb1} to each party locally, the limit bipartite distribution is given by
\begin{align*}
    P(x,y|A,B) = \bra{\psi}  E(x|A) \ox E(y|B) \ket{\psi} \, ,
\end{align*}
where $\ket{\psi} = \sum_{k=0}^{d-1} c_k \ket{k}_A \ox \ket{k}_B \in L^2(\mathbb{R}) \ox L^2(\mathbb{R})$ and $E(x|A) = K(x|A)^\dagger K(x|A)$ with $K(x|A)$ given by \eqref{eq:limitkraus}. This distribution does not admit a local hidden variable model, as we showed by explicit violation of a Bell-CHSH inequality in \cite{gd}.

\begin{figure}[t]
        \centering 
        \includegraphics[width=0.45\textwidth]{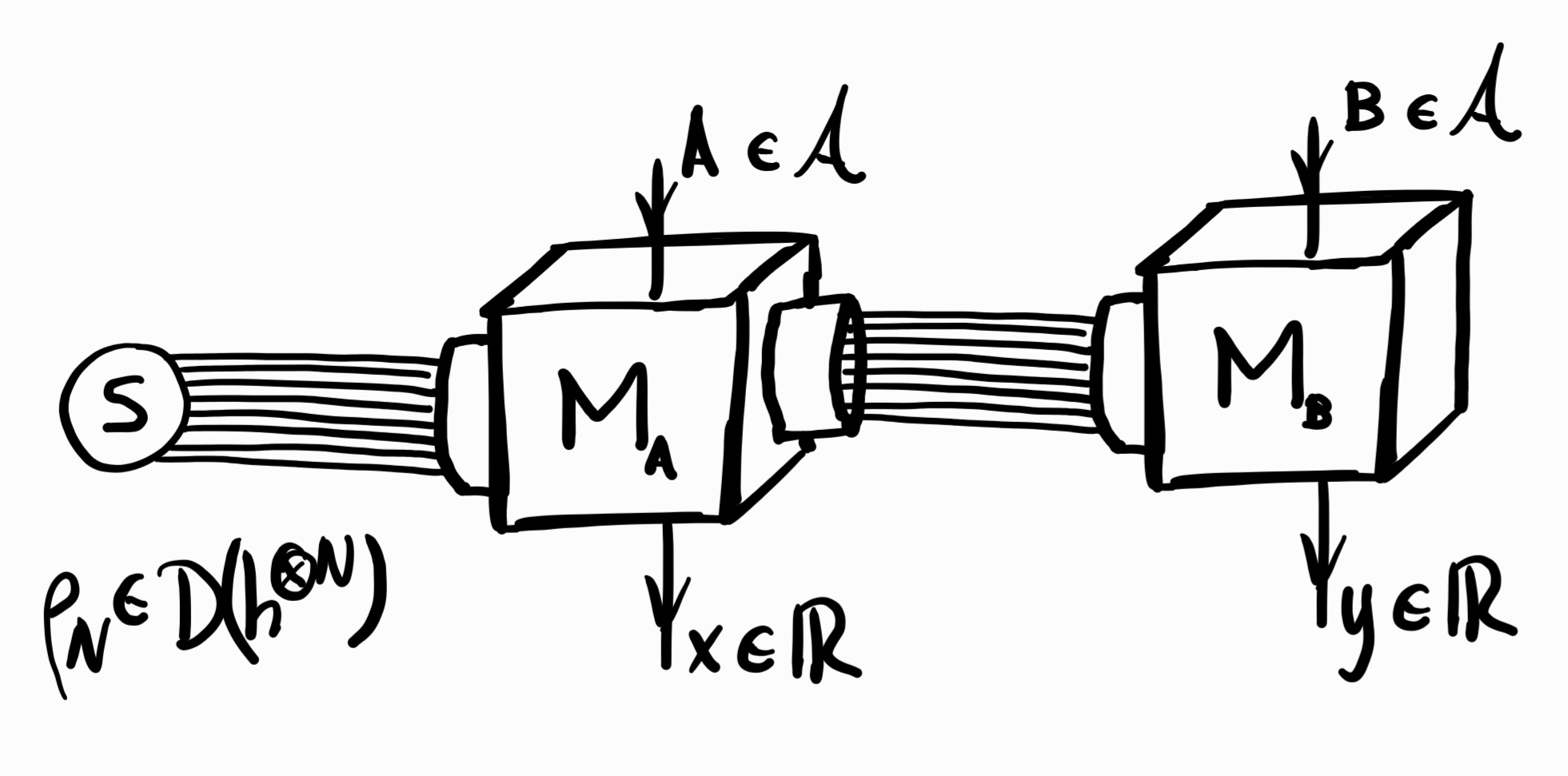} 
        \caption{\textbf{Macroscopic two-time measurement Leggett-Garg experiment.} A macroscopic quantum system $\mathsf{S}$ is consecutively sent to measurement apparatuses $\mathsf{M}_A$ and $\mathsf{M}_B$, which implement coarse-grained measurements.}
        \label{fig:lg}
\end{figure}

\subsubsection*{Violation of a Leggett-Garg inequality}
We now consider a macroscopic Leggett-Garg experiment as depicted in Figure \ref{fig:lg}: a system $\mathsf{S}$, subject to assumptions $\ref{ass:largen}$ - $\ref{ass:loss}$, is consecutively sent into two measurements apparatuses $\mathsf{M}_A$ and $\mathsf{M}_B$, which satisfy conditions \ref{ass:coll} - \ref{ass:cg}. Suppose the system is initially in the state $\sum_{k=0}^{d-1} c_k \ket{N,k} \in \mathfrak{D}_N$. Next, suppose that during the first measurement, defined by a single-particle observable $A \in \mathcal{ND}$, an outcome $x \in \mathbb{R}$ is obtained. Similarly, for the second measurement, we have the associated observable $B \in \mathcal{ND}$ resulting in an outcome $y \in \mathbb{R}$. Then, applying Theorem \ref{thm:mqbn} (namely the MQB$_2$ of the system), the limit bipartite distribution is given by
\begin{align*}
    P(x,y|A,B) = \bra{\psi} K(x|A)^\dagger K(y|B)^\dagger K(y|B) K(x|A) \ket{\psi} \, ,
\end{align*}
where $\ket{\psi} = \sum_{k=0}^{d-1} c_k \ket{k} \in L^2(\mathbb{R})$ and $K(x|A)$ are given by \eqref{eq:limitkraus}. Now consider the following Leggett-Garg CHSH inequality \cite{lgchsh}
\begin{align}
    C = \langle a_1 b_1 \rangle + \langle a_1 b_2 \rangle + \langle a_2 b_1 \rangle - \langle a_2 b_2 \rangle \leq 2  \, ,
\label{eq:lgchsh}
\end{align}
where $a_i = \sgn(x|A_i)$ and  $b_i = \sgn(y|B_i)$. Then, as we show in the Appendix D in the supplementary material, the state 
\begin{align*}
    \ket{\psi} = \sqrt{\frac{1}{2} - \frac{577}{2 \sqrt{1244179}}} \, \ket{0} + \sqrt{\frac{1}{2} + \frac{577}{2 \sqrt{1244179}}} \,  \ket{2} \, .
\end{align*}
gives
\begin{align*}
    C = \frac{2}{675 \pi} \left( 577 + \sqrt{1244179} + 2700 \arctan \frac{1}{3} \right) \simeq 2.42
\end{align*}
for $A_i = \s_x \cos \j_{i} + \s_y \sin \j_{i}$ and $B_i = \s_x \cos \q_{i} + \s_y \sin \q_{i}$ with the following set of angles
\begin{align*}
    \j_{1} = \frac{\pi}{4} \, , \quad \j_{2} = \frac{3 \pi}{4} \, , \quad \q_{1} = \frac{\pi}{2} \, , \quad \q_{2} = 0 \, . 
\end{align*}
This violation of a Leggett-Garg inequality rules out a macroscopic realistic description of the correlations obtained.

\section{Outlook and open questions}
\label{sec:outlook}
Our results shed new light on the question of the quantum-to-classical transition, suggesting that genuine quantum phenomena might be more robust than previously thought. There are several interesting questions to be addressed in the future:
\begin{itemize}
    \item {\bf Question of scale}. Our results show that genuine quantum behaviour can be visible through measurements with a precision of the order of $\sqrt{N}$, even in the presence of (single-particle) decoherence and losses. The relevant parameter that defines the scale of these quantum effects is, therefore, the resolution of the measurements as a function of the system's size $N$. An open question is whether genuine quantum effects exist at a scale larger than $\sqrt{N}$, thus surviving even more coarse-graining than the system we consider. The results of Kofler and Brukner \cite{kb, kb2010} indicate that this is not possible, showing classicality for a resolution much larger than $\sqrt{N}$. But their result only applies to finite-dimensional systems, and the case of infinite dimensional systems (in our language above, the case where the single-particle Hilbert space $\mathfrak{h}$ is infinite-dimensional) is still to be investigated. 

    \item {\bf Macroscopicity measures}. The question of macroscopic quantum states dates back to Schrödinger \cite{schrodinger}, and still today work is done to characterize the ``macroscopicity'' of quantum states (see \cite{frowismacroquantum} and references therein).  An example of a quantum state that is typically thought to be macroscopically quantum is the Greenberger-Horne-Zeilinger (GHZ) \cite{ghz} state $(\ket{0}^{\ox N} + \ket{1}^{\ox N})/2$ (also called cat-like state). However, such a state is extremely fragile, since the loss of coherence of a single particle destroys the coherence of the global state, collapsing it into a classical mixture. For this reason, the GHZ state is not useful for our purposes. Another state typically associated with macroscopic quantumness is the $W$ state \cite{w}, which in our bipartite setting corresponds to the Dicke state $\ket{2N,1} = \frac{1}{\sqrt{2}} \big( \ket{N,1} \ox \ket{N,0} + \ket{N,0} \ox \ket{N,1} \big)$. We have however not found any Bell violation for such a state. These two examples show that quantum macroscopicity does not necessarily result in non-classicality in the macroscopic limit as defined in our sense. More precise relations are to be left for future considerations.
    
    \item {\bf Classical limit and infinite tensor products}. There is a recent proposal that suggests that focusing on type I operator algebras might be the source of the problem and considering instead quantum mechanics on type II operator algebras might provide a framework that encompasses the classical macroscopic limit in a natural way (see \cite{vandenbossche}). This seems to provide another formalism to arrive at the macroscopic limit, and an interesting question to be investigated is how these findings relate to our result.
    
    \item {\bf Experimental implementations}. Of particular interest are experimental considerations to test our findings. Some potential experimental settings that seem to provide the appropriate characteristics are atomic memories and Bose-Einstein condensates (see e.g. \cite{greve, cassens}). While our results are derived in the explicit limit $N\rightarrow\infty$, one has to derive concrete bounds for finite $N$ or at least provide numerical simulations.
\end{itemize}

\begin{acknowledgments}
\emph{Acknowledgments.}--- We would like to thank Joshua Morris for helpful comments. This research was funded in whole, or in part, by the Austrian Science Fund (FWF) [10.55776/F71] and [10.55776/P36994]. M.G. also acknowledges support from the ESQ Discovery programme (Erwin Schrödinger Center for Quantum Science \& Technology), hosted by the Austrian Academy of Sciences (ÖAW). For open access purposes, the author(s) has applied a CC BY public copyright license to any author accepted manuscript version arising from this submission. 
\end{acknowledgments}

\vskip2pc

\bibliographystyle{apsrev4-1}
\bibliography{references}

\onecolumngrid
\appendix

\section{Kraus operators}
\label{app:weakmeasurement}
Let the initial joint state of the system and pointer be described by the density matrix $\r_{\mathsf{S}} \ox \r_{\mathsf{P}}$ where $\r_{\mathsf{S}}$ is the initial state of the system in Fock-like space $\oplus_{M=0}^N \mathfrak{h}^{\ox M}$ and $\r_{\mathsf{P}} = \ket{\Phi}\bra{\Phi}$ is the initial state of the pointer. It is convenient to write $\r_{\mathsf{S}}$ in the eigenbasis of the $A_i$'s, i.e.
\begin{align}
    \r_\mathsf{S} = \bigoplus_{M=0}^N  \sum_{\begin{array}{c} 
    \scriptstyle a_1 \dots a_M \\[-4pt]
    \scriptstyle b_1 \dots b_M
    \end{array}} C_{\begin{array}{c}
    \scriptstyle a_1 \dots a_M \\[-5pt]
    \scriptstyle b_1 \dots b_M
    \end{array}} \, \ket{a_1\dots a_M}_s \, _s\bra{b_1 \dots b_M} \, .
\end{align}
Then, after unitary interaction $U = \exp \left\{ - i \oplus_{M=1}^N \sum_{i=1}^M A_i(s) \ox P \right\}$, the joint state of system and pointer is
\begin{align}
    U  \r_{\mathsf{S}} \ox \r_{\mathsf{P}}   U^\dagger & = \bigoplus_{M=0}^N \sum_{\begin{array}{c} 
    \scriptstyle a_1 \dots a_M \\[-4pt]
    \scriptstyle b_1 \dots b_M
    \end{array}} C_{\begin{array}{c} 
    \scriptstyle a_1 \dots a_M \\[-4pt]
    \scriptstyle b_1 \dots b_M
    \end{array}}   e^{-i \sum_{i=1}^M A_i(s) \ox P} \, \ket{a_1\dots a_M}_s \, _s\bra{b_1 \dots b_M} \ox \ket{\Phi} \bra{\Phi}  \, e^{i \sum_{i=1}^M A_i(s) \ox P} \nonumber \\
    & = \bigoplus_{M=0}^N \sum_{\begin{array}{c} 
    \scriptstyle a_1 \dots a_M \\[-4pt]
    \scriptstyle b_1 \dots b_M
    \end{array}} C_{\begin{array}{c} 
    \scriptstyle a_1 \dots a_M \\[-4pt]
    \scriptstyle b_1 \dots b_M
    \end{array}}   e^{-i \left( \sum_{i=1}^M a_i \right)  P} \, \ket{a_1\dots a_M}_s \, _s\bra{b_1 \dots b_M} \ox \ket{\Phi} \bra{\Phi}  \, e^{i \left( \sum_{i=1}^M b_i \right) P} \nonumber \\
    & = \bigoplus_{M=0}^N \sum_{\begin{array}{c} 
    \scriptstyle a_1 \dots a_M \\[-4pt]
    \scriptstyle b_1 \dots b_M
    \end{array}} C_{\begin{array}{c} 
    \scriptstyle a_1 \dots a_M \\[-4pt]
    \scriptstyle b_1 \dots b_M
    \end{array}}    \ket{a_1\dots a_M}_s \, _s\bra{b_1 \dots b_M} \ox e^{-i \left( \sum_{i=1}^M a_i \right)  P} \, \ket{\Phi} \bra{\Phi}  \, e^{i \left( \sum_{i=1}^M b_i \right) P} \, .
\end{align}
If we measure the position of the pointer obtaining the value $x$, then the state of the system is projected to
\begin{align}
    _\mathsf{P}\mel{x}{ U  \r_{\mathsf{S}}  \ox \r_{\mathsf{P}}    U^\dagger}{x}_\mathsf{P} & = \bigoplus_{M=0}^N \sum_{\begin{array}{c} 
    \scriptstyle a_1 \dots a_M \\[-4pt]
    \scriptstyle b_1 \dots b_M
    \end{array}} C_{\begin{array}{c} 
    \scriptstyle a_1 \dots a_M \\[-4pt]
    \scriptstyle b_1 \dots b_M
    \end{array}}    \ket{a_1\dots a_M}_s \, _s\bra{b_1 \dots b_M} \,  \mel{x}{e^{-i \left( \sum_{i=1}^M a_i \right)  P} \, \ket{\Phi} \bra{\Phi}  \, e^{i \left( \sum_{i=1}^M b_i \right) P} }{x} \nonumber \\
    & = \bigoplus_{M=0}^N \sum_{\begin{array}{c} 
    \scriptstyle a_1 \dots a_M \\[-4pt]
    \scriptstyle b_1 \dots b_M
    \end{array}} C_{\begin{array}{c} 
    \scriptstyle a_1 \dots a_M \\[-4pt]
    \scriptstyle b_1 \dots b_M
    \end{array}}    \ket{a_1\dots a_M}_s \, _s\bra{b_1 \dots b_M} \, \Phi \Big( x - \sum_{i=1}^M a_i \Big) \,  \Phi^* \Big( x - \sum_{i=1}^M b_i \Big) \nonumber \\
    & = K_N(x|s) \,  \r_\mathsf{S} \,K_N(x|s)^\dagger \, ,
\end{align}
where 
\begin{align}
    K_N(x|s) = \bigoplus_{M=0}^N \sum_{a_1 \dots a_M} \Phi \Big( x - \sum_{i=1}^M a_i \Big) \,  \ket{a_1\dots a_M}_s \, _s\bra{b_1 \dots b_M} \, ,
\end{align}
and the probability of obatining such outcome $x$ is $\tr \big[ K_N(x|s) \r_\mathsf{S} K_N^\dagger(x|s) \big]$, proving Equations \eqref{eq:distributionx}, \eqref{eq:kraus} and \eqref{eq:rhon(x,s)} in the main text.

\section{Proof of Theorem \ref{thm:mqb1}.}
\label{app:mqb1}
\begin{proof}
Consider the random variable $X_N$ with distribution
\begin{align}
    P_N(x) = \tr \Big[K_N(x|A)  \L(\r_N) K_N^\dagger (x|A) \Big] \, ,
\end{align}
where $\r_N = \sum_{k,l=0}^{d_N-1} c_{kl} \ket{N,l} \bra{N,k}$, $d_N$ satisfies $\lim_{N \to \infty} d_N/N=0$ and
\begin{align}
    K_N(x|A) = \bigoplus_{M=0}^N \sum_{a_1 \dots a_M}  \Phi \Big( x - \sum_{i=1}^M a_i \Big)  \, \Pi_{a_1 | A}  \otimes \dots \otimes \Pi_{a_M | A}\, .
\end{align}
By Lévy's continuity theorem, in order to show that $X_N$ converges in distribution it is sufficient to show that its characteristic function $\chi_N(t)$
converges pointwise to some function $\chi(t)$ continuous at $t=0$. We have
\begin{align}
    \chi_N(t) & =  \int_{-\infty}^{+\infty} dx \,  e^{itx} \,P_N(x) \nonumber \\
    & =  \int_{-\infty}^{+\infty} dx \,  e^{itx} \, \tr \Big[K(x|A) \L (\r_N ) \, K_N^\dagger (x|A) \Big] \nonumber \\
    & = \tr \left[ \L (\r_N)  \int_{-\infty}^{+\infty} dx \, e^{itx}  K_N^\dagger (x|A) K_N(x|A) \right] \nonumber \\
    & = \tr \left[ \L(\r_N) \bigoplus_{M=0}^N \sum_{a_1 \dots a_M}  \int_{-\infty}^{+\infty} dx \, e^{itx} \Phi^2 \Big( x - \sum_{i=1}^M a_i \Big)  \, \Pi_{a_1 | A}  \otimes \dots \otimes \Pi_{a_M | A} \right] \, .
\end{align}
Using that 
\begin{align}
    \int_{-\infty}^{+\infty} dx \,  e^{itx} \, \Phi^2 \Big( x - \sum_{i=1}^M a_i \Big) & =  \int_{-\infty}^{+\infty} dx \,  e^{itx} \, \frac{e^{-(x- \sum_{i=1}^M a_i)^2/(2 N \s^2)}}{\sqrt{2 \pi N \s^2}} \nonumber \\
    & = e^{i t \sum_{i=1}^M a_i - N \s^2 t^2 /2} \, ,
\end{align}
we have
\begin{align}
    \chi_N(t) & =  e^{- N \s^2 t^2 /2} \, \tr \left[ \L(\r_N) \bigoplus_{M=0}^N \sum_{a_1 \dots a_M}  e^{i t \sum_{i=1}^M a_i}      \, \Pi_{a_1 | A}  \otimes \dots \otimes \Pi_{a_M | A} \right] \nonumber \\
    & = e^{- N \s^2 t^2 /2} \, \tr \left[ \L(\r_N) \, \bigoplus_{M=0}^N \left( \sum_{a_1} e^{i t a_1} \Pi_{a_1|A} \right) \ox \dots \ox \left( \sum_{a_M} e^{i t a_M} \Pi_{a_M|A} \right) \right] \nonumber \\
    & = e^{- N \s^2 t^2 /2} \, \tr \left[ \L(\r_N) \, \bigoplus_{M=0}^N \big( e^{i t A} \big)^{\ox M} \right] \nonumber \\
    & =   e^{- N \s^2 t^2 /2} \,  \tr \left[ \bigoplus_{M=0}^N   {N \choose M} p^{M} (1-p)^{N-M} \, \mathcal{S} \left\{ \tr_{N-M} \Big[ \G^{\otimes N}  (\r_N) \,  \Big] \right\} \big( e^{itA} \big)^{\ox M} \right]  \nonumber \\
    & = e^{- N \s^2 t^2 /2} \, \sum_{M=0}^N {N \choose M} p^{M} (1-p)^{N-M} \tr_M \bigg[\tr_{N-M} \Big[ \G^{\otimes N} (\r_N) \Big] \big( e^{itA} \big)^{\ox M}  \bigg] \nonumber \\
    & = e^{- N \s^2 t^2 /2} \, \sum_{M=0}^N {N \choose M} p^{M} (1-p)^{N-M} \tr_N \left[ \G^{\otimes N} (\r_N) \, \big( e^{itA} \big)^{\ox M} \ox \id^{\ox (N-M)} \right] \nonumber \\
    & =e^{- N \s^2 t^2 /2} \, \tr \left[ \G^{\otimes N} (\r_N) \,  \sum_{M=0}^N {N \choose M} p^{M} (1-p)^{N-M} \, \big( e^{itA} \big)^{\ox M} \ox \id^{\ox (N-M)}  \right] \nonumber \\
    & = e^{- N \s^2 t^2 /2} \, \tr \left[ \G^{\otimes N} (\r_N) \, \big( 1-p + p \, e^{itA} \big)^{\ox N} \right] \nonumber \\
    & = e^{- N \s^2 t^2 /2} \, \tr \left[  \r_N  \,  \G^\dagger \big( 1-p + p \, e^{itA} \big)^{\ox N} \right] \, .
\end{align}
We now perform an affine transformation $X_N \mapsto \l_N X_N + \m_N$ of the random variable with $\l_N = 1/\sqrt{2 N p^2 |G_{01}|^2}$ and $\m_N = -  G_{00} \sqrt{N} / \sqrt{2 |G_{01}|^2 }$, where $G_{ij} := \mel{i}{\Gamma^\dagger(A)}{j}$. The characteristic function of the new random variable is
\begin{align}
    \chi_N(t) & = e^{- i t  G_{00} \sqrt{N}/\sqrt{2 |G_{01}|^2} - \s^2  t^2 / (4 p^2 |G_{01}|^2)} \,  \tr \left[  \r_N  \,  \Gamma^\dagger \Big(1-p+p \, e^{i t A /\sqrt{2 N p^2 |G_{01}|^2} } \Big)^{\ox N} \right] \, .
\end{align}
Defining $\mathcal{G} = \Gamma^\dagger \Big(1-p+p \, e^{i t A /\sqrt{2 N p^2 |G_{01}|^2}} \Big)$, the matrix element $\mel{N,k}{\mathcal{G}^{\otimes N}}{N,l}$ in the case $k \geq l$ is:
\begin{align}
    \bra{N, k}  \, \mathcal{G}^{\ox N} \,  \ket{N, l} & = \frac{1}{\sqrt{ \binom{N}{k} \binom{N}{l}}} \Big( \bbra{\underbrace{1\dots1}_{k} 0\dots0} + \text{perm.} \Big) \mathcal{G}^{\ox N} \Big( \bket{\underbrace{1\dots1}_{l} 0\dots0} +  \text{perm.} \Big) \nonumber \\
    & = \sqrt{\frac{\binom{N}{k}}{ \binom{N}{l}}}  \, \bbra{\underbrace{1\dots1}_{k} 0\dots0}  \, \mathcal{G}^{\ox N} \Big( \bket{\underbrace{1\dots1}_{l} 0\dots0} +  \text{perm.} \Big) \nonumber \\
    & =  \sqrt{\frac{\binom{N}{k}}{ \binom{N}{l}}}  \, \sum_{m=0}^l  \, \binom{k}{m} \binom{N-k}{l-m}  \, \bbra{\underbrace{1\dots1}_{k} 0\dots0}  \, \mathcal{G}^{\ox N} \bket{\underbrace{\underbrace{1\dots1}_{m}0\dots0}_{k} \underbrace{\underbrace{1..1}_{l-m} 0\dots0}_{N-k}}  \nonumber \\
    & = \sqrt{\frac{\binom{N}{k}}{ \binom{N}{l}}}  \,  \sum_{m=0}^l \,  \binom{k}{m} \binom{N-k}{l-m}  \,  \mel{1}{ \mathcal{G} }{1}^{m} \, \mel{1}{\mathcal{G} }{0}^{k-m} \, \mel{0}{ \mathcal{G} }{1}^{l-m} \,  \mel{0}{\mathcal{G} }{0}^{N-k-l+m} \nonumber \\
    & =\sum_{m=0}^l  \frac{\sqrt{k! \, l! \, N^{k+l-2m} \, } }{ m! (k-m)! (l-m)!}\,  \big[ 1 + O(1/N) \big] \,  \mel{1}{ \mathcal{G}}{1}^{m} \, \mel{1}{ \mathcal{G} }{0}^{k-m}  \, \mel{0}{ \mathcal{G} }{1}^{l-m} \,  \mel{0}{\mathcal{G} }{0}^{N-k-l+m} \, .
\label{eq:dickematrixelements}
\end{align}
In the second equality we have used permutational invariance; in the third equality we have gathered all the terms that contribute equally, multiplied by their combinatorial multiplicity; in the last line we have used Stirling's formula. Expanding $\mathcal{G}$ as
\begin{align}
    \mathcal{G} & = \G^\dagger \Big( 1-p + p \, e^{i t A /\sqrt{2 N p^2 |G_{01}|^2}} \Big) \nonumber \\
    & = \G^\dagger \Big( 1 +  p \frac{i t A}{\sqrt{2 N p^2 |G_{01}|^2}} -p  \frac{ t^2 A^2}{4 N p^2 |G_{01}|^2} + O(N^{-3/2}) \Big) \nonumber \\
    & = 1 +  \frac{i t}{\sqrt{2 N |G_{01}|^2}} \G^\dagger (A)-  \frac{ t^2}{4 N p|G_{01}|^2} \G^\dagger(A^2) + O(N^{-3/2})  \, ,
\end{align}
we can see that $\mel{1}{\mathcal{G}}{1}^m = 1 + O(1/\sqrt{N})$ and $\mel{0}{\mathcal{G}}{0}^m = 1 + O(1/\sqrt{N})$, while $\mel{1}{\mathcal{G}}{0}^{k-m} = \Big( \frac{it }{\sqrt{2 N |G_{01}|^2}} \, G_{10} \Big)^{k-m} \big[ 1 + O(1/\sqrt{N}) \big]$ and $\mel{0}{\mathcal{G}}{1}^{l-m} = \Big(  \frac{it}{\sqrt{2 N |G_{01}|^2}} \, G_{01}  \Big)^{l-m} \big[ 1 + O(1/\sqrt{N}) \big]$. These first order contributions of the off-diagonal matrix elements cancel the overall factor of $\sqrt{N^{k+l-2m}}$, while higher order corrections are suppressed. On the other hand,
\begin{align}
    \mel{0}{\mathcal{A}}{0}^{N-k-l} & = \exp \Big\{ (N-k-l) \log \mel{0}{\mathcal{G}}{0} \Big\} \nonumber \\
    & = \exp \bigg\{ (N-k-l) \log \bigg[ 1 + \frac{it }{\sqrt{2N |G_{01}|^2} } \mel{0}{\G^\dagger(A)}{0} - \frac{t^2}{4 N p |G_{01}|^2} \mel{0}{\G^\dagger(A^2)}{0} + O(N^{-3/2}) \bigg] \bigg\} \nonumber \\
    & = \exp \bigg\{ (N-k-l) \bigg[ \frac{i t G_{00}}{\sqrt{2 N |G_{01}|^2} } - \frac{t^2 G^{(2)}_{00}}{4 N p |G_{01}|^2} + \frac{ t^2 G_{00}^2 }{4  N |G_{01}|^2} + O(N^{-3/2}) \bigg] \bigg\} \nonumber \\
    & = \exp \bigg\{ \frac{i t G_{00} \sqrt{N}}{\sqrt{2 |G_{01}|^2}}  - \frac{t^2}{4 p |G_{01}|^2}  \Big( G_{00}^{(2)} - p G_{00}^2   \Big) +O(1/\sqrt{N}) \bigg\} \, .
\end{align}
Then, defining $G = \G^\dagger(A)$ and $G^{(2)} = \G^\dagger(A^2)$, the matrix element reads
\begin{align}
\label{eq:nkannl}
    \mel{N,k}{\mathcal{A}^{\ox N}}{N,l} & = e^{i t G_{00} \sqrt{N} / \sqrt{2 |G_{01}|^2}   -t^2 (G_{00}^{(2)} - p G_{00}^2) / (4 p |G_{01}|^2)  + O(1/\sqrt{N})} \nonumber \\
    & \qquad \qquad  \cdot \sum_{m=0}^l \frac{\sqrt{k! \, l!}}{m! (k-m)! (l-m)!} \left(  \frac{i t G_{10}}{\sqrt{2 |G_{01}|^2}} \right)^{k-m}   \left( \frac{i t G_{01}}{\sqrt{2 |G_{01}|^2}} \right)^{l-m} \big[ 1 +O(1/\sqrt{N}) \big] \, .
\end{align}
For the case $l>k$ all we have to do is exchange $k$ and $l$, so that the sum only runs until the smallest of the two, and also exchange $G_{10}$ and $G_{01}$. Then, defining $\a^2 = \s^2/(p^2 |G_{01}|^2) +  (G^{(2)}_{00} - p G_{00}^2)/(p |G_{01}|^2)$ and $G_{01}=|G_{01}| e^{i \j}$, the characteristic function in the limit reads
\begin{align}
    \lim_{N \to \infty} \chi_N(t) & = e^{-\a^2 t^2/4} \sum_{k,l} e^{i l \j} c_{kl} e^{-ik\j}  \sum_{m=0}^{\min (k,l)}  \frac{\sqrt{k! \, l!}}{m! (k-m)! (l-m)!} \left( \frac{i t}{\sqrt{2}} \right)^{k+l-2m}   \, .
\end{align}
Since this is a continuous characteristic function, we conclude that the random variable $\l_N X_N + \m_N$ converges in distribution. In order to compute its distribution, we take the Fourier transform of the above expression:
\begin{align}
    P(x) & = \frac{1}{2 \pi} \int_{-\infty}^{+\infty} dt \, e^{-itx} \, \chi(t) \nonumber \\
    & =   \sum_{k,l} e^{i l \j} c_{kl} e^{-ik\j}  \sum_{m=0}^{\min (k,l)}  \frac{\sqrt{k! \, l!}}{m! (k-m)! (l-m)!} \left(\frac{1}{\sqrt{2}} \right)^{k+l-2m} \frac{1}{2 \pi}  \int_{-\infty}^{+\infty} dt \,( i t  )^{k+l-2m}  e^{-itx- \a^2 t^2/4}  \nonumber \\
    & =  \sum_{k,l} e^{i l \j} c_{kl} e^{-ik\j} \sum_{m=0}^{\min (k,l)}  \frac{\sqrt{k! \, l!}}{m! (k-m)! (l-m)!} \left( \frac{1}{\sqrt{2}}  \right)^{k+l-2m} (-1)^{k+l-2m} \frac{d^{k+l-2m}}{dx^{k+l-2m}} \frac{1}{2 \pi}   \int_{-\infty}^{+\infty} dt \,  e^{-itx- \a^2 t^2/4}  \nonumber \\
    & =   \sum_{k,l} e^{i l \j} c_{kl} e^{-ik\j}  \sum_{m=0}^{\min (k,l)}  \frac{\sqrt{k! \, l!}}{m! (k-m)! (l-m)!} \left(\frac{1}{\sqrt{2}}  \right)^{k+l-2m} (-1)^{k+l-2m}  \frac{d^{k+l-2m}}{dx^{k+l-2m}} \frac{e^{-x^2/ \a^2}}{\sqrt{\pi  \a^2}}  \, .
\end{align}
Using Rodrigues' formula for Hermite polynomials \cite{abramowitz}
\begin{align}
    H_n(x) = (-1)^n e^{x^2} \frac{d^n}{dx^n} e^{-x^2} \, 
\end{align}
we have
\begin{align}
    P(x) & = \sum_{k,l}  e^{i l \j} c_{kl} e^{-ik\j}  \sum_{m=0}^{\min (k,l)}  \frac{\sqrt{k! \, l!}}{m! (k-m)! (l-m)!} \left( \frac{1}{\sqrt{2}} \right)^{k+l-2m} \frac{e^{- x^2/\a^2}}{ \a^{k+l-2m}} \frac{H_{k+l-2m} ( x/\a)}{\sqrt{\pi  \a^2}} \nonumber \\
    & = \frac{e^{- x^2/\a^2}}{\sqrt{\pi \a^2 }} \sum_{k,l} e^{i l \j} c_{kl} e^{-ik\j} \sqrt{k! \, l!}  \sum_{m=0}^{\min (k,l)}  \frac{1}{m!} {k+l-2m \choose l-m} \left( \frac{1}{\a} \right)^{k+l-2m}  \frac{H_{k+l-2m} ( x/\a)}{\sqrt{2^{k+l-2m}} \, (k+l-2m)!} \, ,
\end{align}
where we have introduced the binomial coefficient for convenience. Now, the above sum over Hermite polynomials can be written as a product of two Hermite polynomials of order $k$ and $l$ by virtue of the following lemma:
\begin{lemma}
For any constants $\a$, $\b$ and $\g$ satisfying $\a^2=\b^2+\g^2$, and for any non-negative integers $k$ and $l$ with $k \geq l$, the following identity of Hermite polynomials holds:
\begin{align}
    \frac{e^{- x^2/\a^2}}{\sqrt{\pi \a^2}} \sum_{m=0}^l \, &  \frac{1}{m!} \, {k+l-2m \choose l-m} \left( \frac{\g}{\a } \right)^{k+l-2m} \, \frac{H_{k+l-2m} \left(x/ \a \right)}{\sqrt{2^{k+l-2m}} \, (k+l-2m)!}  = \nonumber \\
    & \qquad \qquad = \int_{-\infty}^{+\infty} dx'  \frac{e^{-( x-x')^2/\b^2 }}{\sqrt{\pi \b^2}}  \frac{e^{-(x')^2/\g^2}}{\sqrt{\pi \g^2 }} \, \frac{H_k \left( x'/\g \right)}{ \sqrt{2^k} \, k!}  \frac{H_l \left(x'/\g \right)}{\sqrt{2^l} \, l!} \, .
\end{align}
\end{lemma}
\begin{proof}
See \cite{gd}.
\end{proof}
Then, the limit distribution may be written as
\begin{align}
    P(x) & = \sum_{k,l}  e^{i l \j} c_{kl} e^{-ik\j} \sqrt{k! \, l!} \int_{-\infty}^{+\infty} dx'  \frac{e^{-( x-x')^2/\b^2 }}{\sqrt{\pi \b^2}}  \frac{e^{-(x')^2}}{\sqrt{\pi  }} \, \frac{ H_k \left(x' \right)}{ \sqrt{2^k} \, k!}  \frac{ H_l \left(x' \right)}{\sqrt{2^l} \, l!} \nonumber \\
    & = \sum_{k,l}  e^{i l \j} c_{kl} e^{-ik\j}  \int_{-\infty}^{+\infty} dx'  \frac{e^{-( x-x')^2/\b^2 }}{\sqrt{\pi \b^2}}  \left[  \frac{e^{-(x')^2/2}  H_k \left( x'\right) }{\sqrt{2^k \, k! \, \sqrt{\pi }}}   \right] \left[ \frac{e^{-(x')^2/2} H_l \left(x' \right) }{\sqrt{2^l \, l! \, \sqrt{ \pi }}}  \right] \nonumber \\
    & = \sum_{k,l}  e^{i l \j} c_{kl} e^{-ik\j}  \int_{-\infty}^{+\infty} dx'  \frac{e^{-( x-x')^2/\b^2 }}{\sqrt{\pi \b^2}}  \braket{k}{x'} \braket{x'}{l} \, ,
\end{align}
where $\b^2 = \a^2 - 1$ and we have introduced the wave-functions
\begin{align}
    \braket{x}{k} =  \frac{1}{\sqrt{2^k \, k! \sqrt{\pi}}} \, e^{- x^2/2 } \, H_k ( x )
\end{align}
of a one dimensional harmonic oscillator. In conclusion, we can write the limit probability distribution as
\begin{align}
    P(x) = \tr \Big[  K(x|A) \, \r \, K^\dagger(x|A) \Big]  \, ,
\end{align}
where $\r = \sum_{kl} c_{kl} \ket{l} \bra{k}$ and
\begin{align}
    K(x|A) =  \int_{-\infty}^{+\infty} dx' \, \frac{e^{-(x-x')^2/(2 \b^2)}}{(\pi \b^2)^{1/4}} \, \ket{x'}_\j \, _\j\bra{x'} \, , 
\end{align}
where 
\begin{align}
    \b^2 = \frac{\s^2 + p \mel{0}{\G^\dagger(A^2)}{0} - p^2 \mel{0}{\G^\dagger(A)}{0}^2}{p^2 |\mel{0}{\Gamma^\dagger(A)}{1} |^2} - 1
\end{align}
and
\begin{align}
    \j = \arg \mel{0}{\G^\dagger(A)}{1} \, .
\end{align}
\end{proof}

\section{Proof of theorem \ref{thm:mqbn}.}
\label{app:mqbn}
Now consider the random vector $\vec{X}_N$ with distribution
\begin{align}
    P_N(x_1, \dots, x_n | A_1 , \dots , A_n) = \tr \bigg[ K_N(x_n|A_n) \dots K_N(x_1|A_1) \r_N K_N^\dagger(x_1|A_1) \dots K_N^\dagger(x_n|A_n) \bigg] \, .
\end{align}
We want to show that $\vec{X}_N$ converges in distribution to a random vector $\vec{X}$ with distribution
\begin{align}
    P(x_1, \dots, x_n | A_1 , \dots , A_n) = \tr \bigg[ K(x_n|A_n) \dots K(x_1|A_1) \r K^\dagger(x_1|A_1) \dots K^\dagger(x_n|A_n) \bigg] \, .
\end{align}
For this, it is sufficient to show that
\begin{align}
    \limn \langle N , k | \tilde{K}_N^{(1)} \dots \tilde{K}_N^{(n)} \tilde{K}_N^{(n+1)} | N,l \rangle= \sum_{m=0}^\infty \bigg( \limn \langle N,k | \tilde{K}_N^{(1)} \dots \tilde{K}_N^{(n)} | N , m \rangle \bigg) \bigg( \limn \langle N,m | \tilde{K}_N^{(n+1)} | N , l \rangle \bigg) \, ,
\label{eq:limprodprodlim}
\end{align}
where
\begin{align}
    \tilde{K}_N^{(j)} & = \tilde{K}_N(x_j | A_j) \nonumber \\
    & = \frac{1}{\sqrt{\l_N^{(j)}}} \,  K_N \left( \frac{x_j - \m_N^{(j)}}{\l_N^{(j)}}  | A_j \right) \nonumber \\
    & = \frac{1}{\sqrt{\l_N^{(j)}}} \bigoplus_{M=0}^N \sum_{a_1 \dots a_M} \frac{e^{-\big( (x_j - \m_N^{(j)})/\l_N^{(j)} - \sum_{i=1}^M a_i \big)^2/(4 N \s^2) }}{(2 \pi N \s^2)^{1/4}} \Pi_{a_1 | A_j} \ox \dots \ox \Pi_{a_M | A_j} \nonumber \\
    & = \frac{1}{ \big( 2 \pi N \s^2 (\l_N^{(j)})^2 \big)^{1/4}} \bigoplus_{M=0}^N \sum_{a_1 \dots a_M} \sqrt{\frac{ N \s^2 (\l_N^{(j)})^2}{\pi}} \int dp \, e^{i p x_j}  e^{-N \s^2 (\l_N^{(j)})^2 p^2 - i p (\m_N^{(j)} + \l_N^{(j)} \sum_{i=1}^M a_i )} \, \Pi_{a_1 | A_j} \ox \dots \ox \Pi_{a_M | A_j} \nonumber \\
    & = \left( \frac{N \s^2 (\l_N^{(j)})^2 }{2 \pi^3} \right)^{1/4} \int dp \, e^{i p x_j  - N \s^2 (\l_N^{(j)})^2 p^2} \nonumber \\
    & \qquad \qquad \cdot \bigoplus_{M=0}^N e^{-ip\m_N^{(j)} (1-M/N)} \left( \sum_{a_1} e^{-i p ( \l_N^{(j)} a_1 + \m_N^{(j)}/N )} \Pi_{a_1|A_j} \right) \ox \dots \ox \left( \sum_{a_M} e^{-i p (\l_N^{(j)} a_M + \m_N^{(j)}/N)} \Pi_{a_M|A_j} \right) 
\end{align}
for some constants $\l_N^{(j)}$ and $\m_N^{(j)}$. Choosing as before $\l_N^{(j)} = (2 N |(A_j)_{01}|^2)^{-1/2}$ and $\m_N^{(j)} = - (A_j)_{00} \sqrt{N} / \sqrt{2 | (A_j)_{01}|^2}$ and defining $\mathcal{A}_j = \sum_{a} e^{-i p (\l_N^{(j)} a+ \m_N^{(j)} /N)} \Pi_{a|A_j}$, we have
\begin{align}
    \tilde{K}_N^{(j)} & = \left( \frac{\s^2  }{4 \pi^3 |(A_j)_{01}|^2} \right)^{1/4} \int dp \, e^{i p x_j - \s^2 p^2 / (2 |(A_j)_{01}|^2) } \bigoplus_{M=0}^N e^{-ip\m_N^{(j)} (1-M/N)} \mathcal{A}_j^{\ox M} \, .
\end{align}
Then the left hand side of \eqref{eq:limprodprodlim} is
\begin{align}
    \text{LHS} & = \limn \left( \frac{1}{4 \pi^3} \right)^{(n+1)/4} \sqrt{ \frac{\s_1 \dots \s_{n+1}}{ |(A_1)_{01}| \dots |(A_{n+1})_{01}|}} \nonumber \\
    & \qquad \qquad \cdot \int dp_1 \dots dp_{n+1} e^{i \sum_{j=1}^{n+1} x_j p_j  - \frac{1}{2} \sum_{j=1}^{n+1} \s^2 p_j^2 / |(A_j)_{01}|^2 } \langle N,k | ( \mathcal{A}_1 \dots \mathcal{A}_{n+1})^{\ox N} | N, l \rangle \, , 
\end{align}
while the right hand side is
\begin{align}
    \text{RHS} & = \sum_{m=0}^\infty \Bigg\{ \limn \left( \frac{1}{4 \pi^3} \right)^{n/4} \sqrt{ \frac{\s_1 \dots \s_{n}}{ |(A_1)_{01}| \dots |(A_{n})_{01}|}} \nonumber \\
    & \qquad \qquad \qquad \cdot \int dp_1 \dots dp_{n} e^{i \sum_{j=1}^{n} x_j p_j - \frac{1}{2} \sum_{j=1}^{n} \s^2 p_j^2 / |(A_j)_{01}|^2 } \langle N,k | ( \mathcal{A}_1 \dots \mathcal{A}_{n})^{\ox N} | N, m \rangle \Bigg\} \nonumber \\
    &  \qquad \quad \cdot \left\{ \limn \left( \frac{1}{4 \pi^3} \right)^{1/4} \sqrt{ \frac{\s_{n+1}}{ |(A_{n+1})_{01}|}} \int dp_{n+1} e^{i  x_{n+1} p_{n+1} - \frac{1}{2} \s_{n+1}^2 p_{n+1}^2 / |(A_{n+1})_{01}|^2 } \langle N,m |  \mathcal{A}_{n+1}^{\ox N} | N, l \rangle \right\} \, .
\end{align}
It is then sufficient to show that $\text{LHS}' = \text{RHS}'$, where
\begin{align}
    \text{LHS}' = \limn  \langle N,k | ( \mathcal{A}_1 \dots \mathcal{A}_{n+1})^{\ox N} | N, l \rangle 
\end{align}
and
\begin{align}
    \text{RHS}' & = \sum_{m=0}^\infty \bigg( \limn  \langle N,k | ( \mathcal{A}_1 \dots \mathcal{A}_{n})^{\ox N} | N, m \rangle \bigg)  \bigg( \limn \langle N,m |  \mathcal{A}_{n+1}^{\ox N} | N, l \rangle \bigg) \, .
\end{align}
Expanding
\begin{align}
    \mathcal{A}_1 \dots \mathcal{A}_n & = \bigg( 1 - i p_1 \frac{A_1-(A_1)_{00}}{\sqrt{2 N |(A_1)_{01}|^2}} - \frac{p_1^2}{2} \frac{\big( A_1-(A_1)_{00} \big)^2}{2 N |(A_1)_{01}|^2} + O(N^{-3/2}) \bigg) \dots \nonumber \\
    & \qquad \dots \bigg( 1 - i p_n \frac{A_n-(A_n)_{00}}{\sqrt{2 N |(A_n)_{01}|^2}} - \frac{p_n^2}{2} \frac{\big( A_n-(A_n)_{00} \big)^2}{2 N |(A_n)_{01}|^2} + O(N^{-3/2}) \bigg) \nonumber \\
    & = 1 - \frac{i}{\sqrt{2 N}} \sum_{j=1}^n p_j \frac{A_j - (A_j)_{00}}{|(A_j)_{01}|} - \frac{1}{4N} \sum_{j=1}^n p_j^2 \frac{\big( A_j - (A_j)_{00} \big)^2}{|(A_j)_{01}|^2} - \frac{1}{2N} \sum_{i<j}^n p_i p_j \frac{A_i - (A_i)_{00}}{|(A_i)_{01}|} \frac{A_j - (A_j)_{00}}{|(A_j)_{01}|} + O(N^{-3/2}) \, , 
 \end{align}
we have that 
\begin{align}
    \langle 1 | \mathcal{A}_1 \dots \mathcal{A}_n |1 \rangle & = 1 + O(N^{-1/2}) \, , \nonumber \\
    \langle 1 | \mathcal{A}_1 \dots \mathcal{A}_n |0 \rangle & = - \frac{i}{\sqrt{2N}} \sum_{j=1}^n p_j e^{-i \j_j}+ O(N^{-1}) \, ,\nonumber \\
    \langle 0 | \mathcal{A}_1 \dots \mathcal{A}_n |1 \rangle & = - \frac{i}{\sqrt{2N}} \sum_{j=1}^n p_j e^{i \j_j}+ O(N^{-1}) \, ,\nonumber \\
    \langle 0 | \mathcal{A}_1 \dots \mathcal{A}_n |0 \rangle & = 1 - \frac{1}{4N}  \sum_{j=1}^n p_j^2 - \frac{1}{2N} \sum_{i<j}^n p_i p_j e^{i (\j_i - \j_j)}+ O(N^{-3/2}) \, ,
\end{align}
where we have defined $\langle 0 | A_j | 1 \rangle = |(A_j)_{01}| e^{i \j_j}$, so that using \eqref{eq:dickematrixelements} (from the proof in the previous appendix) we have
\begin{align}
    \limn \langle N,k | \mathcal{A}_1 \dots \mathcal{A}_n | N,l \rangle & = e^{-\frac{1}{4} \sum_{j=1}^n p_j^2 - \frac{1}{2} \sum_{i<j}^n p_i p_j e^{i(\j_i-\j_j)}} \nonumber \\
    & \qquad \cdot \sum_{r=0}^{\min(k,l)} \frac{\sqrt{k! l!}}{r! (k-r)! (l-r)!} \left( -\frac{i}{\sqrt{2}} \sum_{j=1}^n p_j e^{-i \j_j} \right)^{k-r} \left( -\frac{i}{\sqrt{2}} \sum_{j=1}^n p_j e^{i \j_j} \right)^{l-r} \, .
\end{align}
Therefore, defining
\begin{align}
    a & = - \frac{i}{\sqrt{2}} \sum_{j=1}^n p_j e^{i \j_j} \, , \nonumber \\
    b & = - \frac{i}{\sqrt{2}} \sum_{j=1}^n p_j e^{-i \j_j} \, ,\nonumber \\
    a' & = - \frac{i}{\sqrt{2}}  p_{n+1} e^{i \j_{n+1}} \, , \nonumber \\
    b & = - \frac{i}{\sqrt{2}} p_{n+1} e^{-i \j_{n+1}} \, ,
\end{align}
we have
\begin{align}
    \text{LHS}' & = e^{-\frac{1}{4} \sum_{j=1}^{n+1} p_j^2- \frac{1}{2} \sum_{i<j}^{n+1} p_i p_j e^{i(\j_i-\j_j)}} \sum_{r=0}^{\min(k,l)} \frac{\sqrt{k! l!}}{r! (k-r)! (l-r)!} \left( b+b'\right)^{k-r} \left( a+a' \right)^{l-r} \nonumber \\
    & = \sqrt{k! \, l!} \, e^{-\frac{1}{4} \sum_{j=1}^{n+1} p_j^2- \frac{1}{2} \sum_{i<j}^{n} p_i p_j e^{i(\j_i-\j_j)}} f_{kl}(a,b,a',b') \, ,
\end{align}
where we have defined the function
\begin{align}
    f_{kl} (a, b, a', b') & = e^{ab'} (a+a')^l (b+b')^k \sum_{r=0}^{\min(k,l)} \frac{ (a + a')^{-r} ( b+b')^{-r}}{r! (k-r)! (l-r)!} \nonumber \\
    & = e^{ab'} (a+a')^{l-\min(k,l)} \, (b+b')^{k-\min(k,l)} \sum_{r=0}^{\min(k,l)} \frac{ (a + a')^{\min(k,l)-r} \, ( b+b')^{\min(k,l)-r}}{r! (k-r)! (l-r)!} \nonumber \\
    & = e^{ab'} (a+a')^{l-\min(k,l)} \, (b+b')^{k-\min(k,l)} \sum_{q=0}^{\min(k,l)} \frac{ (a + a')^{q} \, ( b+b')^{q}}{ \big( \min(k,l)-q \big)! \, \big(\max(k,l)- \min(k,l)+q \big)! \, q!} \nonumber \\
    & = e^{ab'} \frac{(a+a')^{l-\min(k,l)} \, (b+b')^{k-\min(k,l)}}{\max(k,l)!} \sum_{q=0}^{\min(k,l)} {\max(k,l) \choose \min(k,l)-q} \frac{ (a + a')^{q} \, ( b+b')^{q}}{ q!} \nonumber \\
     & =  e^{ab'} \frac{(a+a')^{l- \min(k,l)} (b+b')^{k-\min(k,l)}}{\max(k,l)!} L_{\min(k,l)}^{|k-l|} \Big(-(a + a') ( b+b') \Big) \, ,
\end{align}
in terms of the generalized Laguerre polynomials
\begin{align}
    L_k^m(x) = \sum_{q=0}^k {k+m \choose k-q} \frac{(-x)^q}{q!} \, .
\end{align}
On the other hand, we have that
\begin{align}
    \text{RHS}' & =  \sum_{m=0}^\infty \left\{ e^{-\frac{1}{4} \sum_{j=1}^{n} p_j^2- \frac{1}{2} \sum_{i<j}^n p_i p_j e^{i(\j_i-\j_j)}} \sum_{s=0}^{\min(k,m)}  \frac{\sqrt{k! m!}}{s! (k-s)! (m-s)!} b^{k-s} a^{m-s}\right\} \nonumber\\
    & \qquad \qquad \cdot \left\{ e^{-\frac{1}{4}  p_{n+1}^2} \sum_{t=0}^{\min(m,l)}  \frac{\sqrt{m! l!}}{t! (m-t)! (l-t)!} b'^{m-t} a'^{l-t}\right\} \nonumber \\
    & = \sqrt{k! \, l!} \, e^{-\frac{1}{4} \sum_{j=1}^{n+1} p_j^2- \frac{1}{2} \sum_{i<j}^n p_i p_j e^{i(\j_i-\j_j)}} g_{kl}(a,b,a',b')
\end{align}
where we have defined the function
\begin{align}
    g_{kl}(a, b, a', b') & = (a')^l b^k \sum_{m=0}^\infty m! \, (ab')^m \sum_{s=0}^{\min(k,m)} \frac{(a b)^{-s}}{s! (k-s)! (m-s)!} \sum_{t=0}^{\min(m,l)} \frac{(a'b')^{-t}}{t! (m-t)! (l-t)!} \nonumber \\
    & = (a')^l b^k \sum_{m=0}^\infty m! \, (ab')^m \, (ab)^{-\min(k,m)}  \sum_{s=0}^{\min(k,m)}  \frac{(ab)^{\min(k,m) -s}}{s! (k-s)! (m-s)!} \, (a' b')^{-\min(m,l)}  \sum_{t=0}^{\min(m,l)} \frac{(a'b')^{\min(m,l)-t}}{t! (m-t)! (l-t)!}  \nonumber \\
    & = (a')^l b^k \sum_{m=0}^\infty m! \, (ab')^m \, (ab)^{-\min(k,m)}  \sum_{u=0}^{\min(k,m)}  \frac{(ab)^{u}}{\big( \min(k,m)-u \big)!  \, \big( \max(k,m) - \min(k,m) + u \big)! \, u!} \nonumber \\
    & \qquad \qquad \qquad \qquad \qquad \cdot   \, (a' b')^{-\min(m,l)}  \sum_{v=0}^{\min(m,l)} \frac{(a'b')^{v}}{\big( \min(m,l)-v \big)!  \, \big( \max(m,l) - \min(m,l) + v \big)! \, v!}  \nonumber \\
     & = (a')^l b^k \sum_{m=0}^\infty m! \, (ab')^m \, \frac{(ab)^{-\min(k,m)}}{\max(k,m)!} \sum_{u=0}^{\min(k,m)} {\max(k,m) \choose \min(k,m)-u} \frac{(ab)^{u}}{u!} \nonumber \\
    & \qquad \qquad \qquad \qquad \qquad \cdot  \frac{(a' b')^{-\min(m,l)}}{\max(m,l)!} \sum_{v=0}^{\min(m,l)} {\max(m,l) \choose \min(m,l)-v} \frac{(a'b')^{v}}{v!} \nonumber \\
    & = (a')^l b^k \sum_{m=0}^\infty m! \, (ab')^m \, \frac{(ab)^{-\min(k,m)}}{\max(k,m)!} L_{\min(k,m)}^{|k-m|} (-ab) \, \frac{(a' b')^{-\min(m,l)}}{\max(m,l)!} L_{\min(m,l)}^{|m-l|} (-a'b') \, .
\end{align}
Therefore a sufficient condition for the identity $\text{LHS}=\text{RHS}$ is that $f_{kl}(a,b,a',b') = g_{kl}(a,b,a',b')$ for all $k,l \geq 0$ and for all $a, b, a', b' \in \mathbb{C}$, which holds by virtue of the following lemma:
\begin{lemma}
    Let
    \begin{align}
        f_{kl}(a,b,a',b') & = e^{ab'} \frac{(a+a')^{l-\min(k,l)} \, (b+b')^{k-\min(k,l)}}{\max(k,l)!} L_{\min(k,l)}^{|k-l|} \big( - (a+a') (b+b') \big)
    \end{align}
    and
    \begin{align}
        g_{kl} (a,b,a',b') & = (a')^l b^k \sum_{j=0}^\infty j! (ab')^j \frac{(ab)^{-\min(k,j)}}{\max(k,j)!} L_{\min(k,j)}^{|k-j|} (-ab) \frac{(a'b')^{-\min(l,j)}}{\max(l,j)!} L_{\min(l,j)}^{|l-j|} (-a'b') \, .
    \end{align}
    Then $f_{kl}(a,b,a',b') = g_{kl}(a,b,a',b')$ for every non-negative integers $k$ and $l$ and for every $a,a',b,b' \in \mathbb{C}$. 
\label{lemma}
\end{lemma}
\begin{proof}
    The structure of the proof is as follows. First, seeing that $f_{00}=g_{00}$ is easy. Then we prove that, in the case $l=0$, both $f$ and $g$ satisfy the same recurrence relation, namely
    \begin{align}
        f_{k+1,0} = \frac{b+b'}{k+1} f_{k0} \qquad \text{and} \qquad g_{k+1,0} = \frac{b+b'}{k+1} g_{k0} \, ,
    \tag{REC1}
    \label{eq:rec1}
    \end{align}
    which by induction on $k$ implies that $f_{k0}=g_{k0}$ for all $k \geq 0$. Next, for $l\neq 0$, we prove that, in the case $k \geq l+1$, the following recurrence relations are satisfied:
    \begin{align}
        f_{k,l+1}=\frac{a+a'}{l+1} f_{kl} + \frac{1}{l+1} f_{k-1,l} \qquad \text{and} \qquad g_{k,l+1}=\frac{a+a'}{l+1} g_{kl} + \frac{1}{l+1} g_{k-1,l} \, .
    \tag{REC2}
    \label{eq:rec2}
    \end{align}
    This, together with the previously established identity $f_{k0} = g_{k0}$, implies by induction on $l$ that $f_{k,l+1}=g_{k,l+1}$ for $k \geq l+1$, so that we have $f_{kl}=g_{kl}$ for $k\geq l$. Finally, since the identity is symmetric with respect to the exchange of $k$ and $l$, $a$ and $b$ and $a'$ and $b'$, the identity is proven for all $k$ and $l$. Therefore, it remains to prove equations \eqref{eq:rec1} and \eqref{eq:rec2}. Before proceeding further, let us recall some useful identities of Laguerre polynomials:
    \begin{align}
    L_n^m(x) & = L_n^{m+1}(x) - L_{n-1}^{m+1}(x) \, ,
    \tag{$\ast$}
    \label{eq:Ln}\\
    n L_n^m(x) & = (n+m) L_{n-1}^m(x) - x L_{n-1}^{m+1}(x) \, , 
    \tag{$\ast\ast$}
    \label{eq:nLn}
    \end{align}
    (with the convention that $L_n^m(x) = 0$ if $n < 0 $).\\
    
    \underline{Step 1:} proof of \eqref{eq:rec1}. For $l=0$ we have
    \begin{align}
        f_{k0} & = e^{ab'} \frac{(b+b')^k}{k!} \qquad \text{and} \qquad g_{k0}  = b^k \sum_{j=0}^\infty (ab')^j \frac{(ab)^{-\min(k,j)}}{\max(k,j)!} L_{\min(k,j)}^{|k-j|}(-ab) \, .
    \end{align}
    It is easy to see that the recurrence relation $f_{k+1,0} = \frac{b+b'}{k+1} f_{k0}$ is satisfied. Similarly,
    \begin{align}
        g_{k+1,0} & = b^{k+1} \left\{ \sum_{j=0}^k (ab')^j \frac{(ab)^{-j}}{(k+1)!} L_j^{k+1-j}(-ab) + \sum_{j=k+1}^\infty (ab')^j \frac{(ab)^{-k-1}}{j!} L_{k+1}^{j-k-1} (-ab)  \right\} \nonumber \\
        & = b^{k+1} \left\{ \frac{1}{(k+1)!} \sum_{j=0}^k \left( \frac{b'}{b} \right)^j  L_j^{k+1-j}(-ab) + \sum_{j=k+1}^\infty (ab')^j \frac{(ab)^{-k-1}}{j!} \, \frac{j L_k^{j-k-1}(-ab) + ab L_k^{j-k}(-ab)}{k+1}  \right\} \nonumber \\
        & = \frac{b^{k+1}}{k+1} \Bigg\{  \frac{1}{k!} \sum_{j=0}^k \left( \frac{b'}{b} \right)^j L_j^{k+1-j}(-ab)   \nonumber \\
        & \qquad \qquad \qquad   + ab' (ab)^{-1} \sum_{j=k+1}^\infty (ab')^{j-1} \frac{(ab)^{-k}}{(j-1)!} L_{k}^{j-k-1} (-ab)  + \sum_{j=k+1}^\infty (ab')^{j} \frac{(ab)^{-k}}{j!} L_{k}^{j-k} (-ab) \Bigg\} \nonumber \\
        & = \frac{b^{k+1}}{k+1} \Bigg\{  \frac{1}{k!} \sum_{j=0}^k \left( \frac{b'}{b} \right)^j L_j^{k+1-j}(-ab)   \nonumber \\
        & \qquad \qquad \qquad + \frac{b'}{b} \sum_{j=k}^\infty (ab')^{j} \frac{(ab)^{-k}}{j!} L_{k}^{j-k} (-ab) + \frac{g_{k0}}{b^k} - \sum_{j=0}^k (ab')^{j} \frac{(ab)^{-j}}{k!} L_{j}^{k-j} (-ab) \Bigg\} \nonumber \\
        & = \frac{b}{k+1} g_{k0} +  \frac{b^{k+1}}{k+1} \Bigg\{  \frac{1}{k!} \sum_{j=0}^k \left( \frac{b'}{b} \right)^j L_j^{k+1-j}(-ab)  \nonumber \\
         & \qquad \qquad \qquad \qquad \qquad \qquad + \frac{b'}{b} \bigg[ \frac{g_{k0}}{b^k} - \sum_{j=0}^{k-1} (ab')^j \frac{(ab)^{-j}}{k!} L_j^{k-j}(-ab) \bigg]  - \frac{1}{k!} \sum_{j=0}^k \left( \frac{b'}{b} \right)^j L_{j}^{k-j} (-ab) \Bigg\} \nonumber \\
        & = \frac{b+b'}{k+1} g_{k0} +  \frac{b^{k+1}}{(k+1)!} \Bigg\{  \sum_{j=0}^k \left( \frac{b'}{b} \right)^j L_j^{k+1-j}(-ab)  \nonumber \\
        & \qquad \qquad \qquad \qquad \qquad \qquad - \frac{b'}{b}  \sum_{j=0}^{k-1} \left( \frac{b'}{b} \right)^j L_j^{k-j}(-ab) - \sum_{j=0}^k \left( \frac{b'}{b} \right)^j L_{j}^{k-j} (-ab) \Bigg\}   \nonumber \\
        & =: \frac{b+b'}{k+1} g_{k0} +  \frac{b^{k+1}}{(k+1)!} R \, ,
    \end{align}
    where in the second equality we have used \eqref{eq:nLn} and in the last equality we have defined
    \begin{align}
        R & := \sum_{j=0}^k \left( \frac{b'}{b} \right)^j \left[  L_j^{k+1-j}(-ab) - L_j^{k-j}(-ab) \right] - \frac{b'}{b} \sum_{j=0}^{k-1} \left( \frac{b'}{b} \right)^j L_j^{k-j}(-ab) \nonumber \\ 
        & = \sum_{j=1}^k \left( \frac{b'}{b} \right)^j  L_{j-1}^{k+1-j}(-ab)  -  \sum_{j=0}^{k-1} \left( \frac{b'}{b} \right)^{j+1} L_j^{k-j}(-ab) \nonumber \\
        & = 0 \, ,
    \end{align}
where in the second equality we have used \eqref{eq:Ln}. \\

\underline{Step 2:} proof of \eqref{eq:rec2}. For $l \neq 0$ and $k\geq l$ we have
\begin{align}
    f_{kl} & = e^{ab'} \frac{ (b+b')^{k-l}}{k!} L_{l}^{k-l} \big(-(a+a')(b+b')\big) \, , 
\end{align}
and for $k \geq l+1$ we have the following recurrence relation for $f$:
\begin{align}
    f_{k,l+1} & = e^{ab'} \frac{(b+b')^{k-l-1}}{k!} L_{l+1}^{k-l-1} \big(-(a+a')(b+b')\big) \nonumber \\
    & = e^{ab'} \frac{(b+b')^{k-l-1}}{k!} \, \frac{k L_l^{k-l-1} \big(-(a+a')(b+b')\big) + (a+a')(b+b') L_l^{k-l} \big(-(a+a')(b+b')\big)}{l+1} \nonumber \\
    & = \frac{1}{l+1} e^{ab'} \frac{(b+b')^{k-l-1}}{(k-1)!} L_l^{k-l-1}  \big(-(a+a')(b+b')\big) + \frac{a+a'}{l+1} e^{ab'} \frac{(b+b')^{k-l}}{k!} L_l^{k-l}  \big(-(a+a')(b+b')\big) \nonumber \\
    & = \frac{1}{l+1} f_{k-1,l} + \frac{a+a'}{l+1} f_{kl} \, ,
\end{align}
where in the second equality we have used \eqref{eq:nLn}. For $g$ we have
\begin{align}
    g_{k,l+1} & = (a')^{l+1} b^k \sum_{j=0}^\infty j! (ab')^j \frac{(ab)^{-\min(k,j)}}{\max(k,j)!} L_{\min(k,j)}^{|k-j|} (-ab) \frac{(a'b')^{-\min(l+1,j)}}{\max(l+1,j)!} L_{\min(l+1,j)}^{|l+1-j|}(-a'b') \nonumber \\
    & = (a')^{l+1} b^k \Bigg\{ \sum_{j=0}^l j! (ab')^j \frac{(ab)^{-j}}{k!} L_{j}^{k-j} (-ab) \frac{(a'b')^{-j}}{(l+1)!} L_{j}^{l+1-j}(-a'b')  \nonumber \\
    & \qquad \qquad \qquad \qquad  + \sum_{j=l+1}^\infty  j! (ab')^j \frac{(ab)^{-\min(k,j)}}{\max(k,j)!} L_{\min(k,j)}^{|k-j|} (-ab) \frac{(a'b')^{-l-1}}{j!} L_{l+1}^{j-l-1}(-a'b') \Bigg\}  \nonumber \\
    & = (a')^{l+1} b^k \Bigg\{ \frac{1}{k! \, (l+1)!} \sum_{j=0}^l j!  (a'b)^{-j} L_{j}^{k-j} (-ab)  L_{j}^{l+1-j}(-a'b')  \nonumber \\
    & \qquad \qquad \qquad \quad +  \sum_{j=l+1}^\infty j! (ab')^j \frac{(ab)^{-\min(k,j)}}{\max(k,j)!} L_{\min(k,j)}^{|k-j|} (-ab)  \frac{(a'b')^{-l-1}}{j!} \, \frac{j L_l^{j-l-1} (-a' b') + a' b' L_l^{j-l}(-a'b')}{l+1} \Bigg\}  \nonumber \\
    & = \frac{(a')^{l+1} b^k}{l+1} \Bigg\{ \frac{1}{k! \, l!} \sum_{j=0}^l j! (a'b)^{-j} L_{j}^{k-j} (-ab)  L_{j}^{l+1-j}(-a'b')  \nonumber \\
    & \qquad \qquad \qquad \quad + (a'b')^{-l-1}\sum_{j=l+1}^\infty   (ab')^j \frac{(ab)^{-\min(k,j)}}{\max(k,j)!} L_{\min(k,j)}^{|k-j|} (-ab)  \,  j \,  L_l^{j-l-1} (-a' b')   \nonumber \\
    & \qquad \qquad \qquad \quad +  \sum_{j=l+1}^\infty  j! (ab')^j \frac{(ab)^{-\min(k,j)}}{\max(k,j)!} L_{\min(k,j)}^{|k-j|} (-ab)  \frac{(a' b')^{-l}}{j!}   L_l^{j-l}(-a'b')\Bigg\}  \nonumber \\
    & = \frac{(a')^{l+1} b^k}{l+1} \Bigg\{ \frac{1}{k! \, l!} \sum_{j=0}^l j! (a'b)^{-j}  L_{j}^{k-j} (-ab) L_{j}^{l+1-j}(-a'b')  \nonumber \\
    & \qquad \qquad \qquad \quad + (a'b')^{-l-1}  \sum_{j=l+1}^k  (ab')^j \frac{(ab)^{-j}}{k!} L_{j}^{k-j} (-ab) \,  j \,  L_l^{j-l-1} (-a' b')   \nonumber \\
    & \qquad \qquad \qquad \quad + (a'b')^{-l-1} \sum_{j=k+1}^\infty   (ab')^j \frac{(ab)^{-k}}{j!} L_{k}^{j-k} (-ab) \, j \, L_l^{j-l-1} (-a' b')   \nonumber \\
    & \qquad \qquad \qquad \quad + \frac{g_{kl}}{(a')^l b^k} - \sum_{j=0}^l j! (ab')^j \frac{(ab)^{-j}}{k!} L_{j}^{k-j} (-ab) \frac{(a'b')^{-j}}{l!} L_j^{l-j}(-a'b') \Bigg\} \nonumber 
\end{align}
\begin{align}
    & = \frac{a'}{l+1} g_{kl} + \frac{(a')^{l+1} b^k}{k! \, (l+1)!} \sum_{j=0}^l j! (a'b)^{-j}  L_{j}^{k-j} (-ab)  \left( L_{j}^{l+1-j}(-a'b')  - L_j^{l-j}(-a'b') \right) \nonumber \\
    & \qquad \qquad \quad + \frac{b^k (b')^{-l-1}}{k! \, (l+1)} \sum_{j=l+1}^k j \left( \frac{b'}{b} \right)^{j} L_j^{k-j} (-ab) L_l^{j-l-1}(-a' b') \nonumber \\
    & \qquad \qquad \quad +\frac{a^{-k+1} (b')^{-l}}{l+1} \sum_{j=k+1}^\infty \frac{(ab')^{j-1}}{(j-1)!} L_k^{j-k}(-ab) L_l^{j-l-1}(-a'b') \nonumber \\
    & = \frac{a'}{l+1} g_{kl} + \frac{(a')^{l+1} b^k}{k! \, (l+1)!} \sum_{j=1}^l j! (a'b)^{-j}  L_{j}^{k-j} (-ab) L_{j-1}^{l+1-j}(-a'b') \nonumber \\
    & \qquad \qquad \quad + \frac{b^k (b')^{-l-1}}{k! \, (l+1)} \sum_{j=l+1}^k  \left( \frac{b'}{b} \right)^{j} j L_j^{k-j} (-ab) L_l^{j-l-1}(-a' b') \nonumber \\
    & \qquad \qquad \quad +\frac{a^{-k+1} (b')^{-l}}{l+1}  \sum_{j=k}^\infty \frac{(ab')^j}{j!} L_k^{j+1-k}(-ab) L_l^{j-l}(-a'b')    \nonumber \\
    & = \frac{a'}{l+1} g_{kl} + \frac{(a')^{l+1} b^k}{k! \, (l+1)!} \sum_{j=1}^l j! (a'b)^{-j}  L_{j}^{k-j} (-ab)  L_{j-1}^{l+1-j}(-a'b') \nonumber \\
    & \qquad \qquad \quad + \frac{b^k (b')^{-l-1}}{k! \, (l+1)} \sum_{j=l+1}^k  \left( \frac{b'}{b} \right)^{j}  \Big( k L_{j-1}^{k-j}(-ab) + ab L_{j-1}^{k-j+1}(-ab) \Big) L_l^{j-l-1}(-a' b') \nonumber \\
    & \qquad \qquad \quad +\frac{a^{-k+1} (b')^{-l}}{l+1} \sum_{j=k}^\infty \frac{(ab')^j}{j!} \Big( L_k^{j-k} (-ab) + L_{k-1}^{j-k+1} (-ab) \Big) L_l^{j-l}(-a'b')    \nonumber \\
    & = \frac{a'}{l+1} g_{kl} + \frac{(a')^{l+1} b^k}{k! \, (l+1)!} \sum_{j=1}^l j! (a'b)^{-j} L_{j}^{k-j} (-ab)  L_{j-1}^{l+1-j}(-a'b') \nonumber \\
    & \qquad \qquad \quad + \frac{b^k (b')^{-l-1}}{(k-1)! \, (l+1)} \sum_{j=l+1}^k  \left( \frac{b'}{b} \right)^{j}  L_{j-1}^{k-j}(-ab)  L_l^{j-l-1}(-a' b') \nonumber \\
    & \qquad \qquad \quad + \frac{a b^{k+1} (b')^{-l-1}}{k! \, (l+1)} \sum_{j=l+1}^k  \left( \frac{b'}{b} \right)^{j}  L_{j-1}^{k-j+1}(-ab) L_l^{j-l-1}(-a' b') \nonumber \\
    & \qquad \qquad \quad +\frac{a^{-k+1} (b')^{-l}}{l+1} \sum_{j=k}^\infty \frac{(ab')^j}{j!}  L_k^{j-k} (-ab) L_l^{j-l}(-a'b')    \nonumber \\
    & \qquad \qquad \quad +\frac{a^{-k+1} (b')^{-l}}{l+1} \sum_{j=k}^\infty \frac{(ab')^j}{j!}  L_{k-1}^{j-k+1} (-ab)  L_l^{j-l}(-a'b') \, ,
\end{align}
where we have used \eqref{eq:nLn} in the third and eighth equalities and \eqref{eq:Ln} in the seventh and eighth. On the other hand we have 
\begin{align}
    g_{kl} & = (a')^l b^k \sum_{j=0}^\infty j! (ab')^j \frac{(ab)^{-\min(k,j)}}{\max(k,j)!} L_{\min(k,j)}^{|k-j|} (-ab) \frac{(a'b')^{-\min(l,j)}}{\max(l,j)!} L_{\min(l,j)}^{|l-j|} (-a'b') \nonumber \\
    & = (a')^l b^k \sum_{j=0}^{k-1} j! (ab')^j \frac{(ab)^{-j}}{k!} L_{j}^{k-j} (-ab) \frac{(a'b')^{-\min(l,j)}}{\max(l,j)!} L_{\min(l,j)}^{|l-j|} (-a'b') \nonumber \\
    & \quad \, + (a')^l b^k \sum_{j=k}^\infty j! (ab')^j \frac{(ab)^{-k}}{j!} L_{k}^{j-k} (-ab) \frac{(a'b')^{-l}}{j!} L_{l}^{j-l} (-a'b') \nonumber \\
    & = \frac{ (a')^l b^k }{k!} \sum_{j=0}^{k-1} j! \left( \frac{b'}{b} \right)^j  L_{j}^{k-j} (-ab) \frac{(a'b')^{-\min(l,j)}}{\max(l,j)!} L_{\min(l,j)}^{|l-j|} (-a'b') \nonumber \\
    & \quad \, + a^{-k} (b')^{-l} \sum_{j=k}^\infty \frac{(ab')^j}{j!} L_{k}^{j-k} (-ab) L_{l}^{j-l} (-a'b')
\end{align}
and 
\begin{align}
    g_{k-1,l} & = (a')^l b^{k-1} \sum_{j=0}^\infty j! (ab')^j \frac{(ab)^{-\min(k-1,j)}}{\max(k-1,j)!} L_{\min(k-1,j)}^{|k-1-j|} (-ab) \frac{(a'b')^{-\min(l,j)}}{\max(l,j)!} L_{\min(l,j)}^{|l-j|} (-a'b')  \nonumber \\
    & = (a')^l b^{k-1} \sum_{j=0}^{k-1} j! (ab')^j \frac{(ab)^{-j}}{(k-1)!} L_{j}^{k-1-j} (-ab) \frac{(a'b')^{-\min(l,j)}}{\max(l,j)!} L_{\min(l,j)}^{|l-j|} (-a'b')  \nonumber \\
    & \quad \, +  (a')^l b^{k-1} \sum_{j=k}^\infty j! (ab')^j \frac{(ab)^{-k+1}}{j!} L_{k-1}^{j-k+1} (-ab) \frac{(a'b')^{-l}}{j!} L_{l}^{j-l} (-a'b')  \nonumber \\
    & = \frac{(a')^l b^{k-1} }{(k-1)!} \sum_{j=0}^{k-1} j! \left( \frac{b'}{b} \right)^j  L_{j}^{k-1-j} (-ab) \frac{(a'b')^{-\min(l,j)}}{\max(l,j)!} L_{\min(l,j)}^{|l-j|} (-a'b')  \nonumber \\
    & \quad \, +  a^{-k+1} (b')^{-l} \sum_{j=k}^\infty \frac{(ab')^j}{j!} L_{k-1}^{j-k+1} (-ab)  L_{l}^{j-l} (-a'b')  \, .
\end{align}
Therefore,
\begin{align}
    g_{k,l+1} & = \frac{a'}{l+1} g_{kl} + \frac{(a')^{l+1} b^k}{k! \, (l+1)!} \sum_{j=1}^l j! (a'b)^{-j} L_{j}^{k-j} (-ab)  L_{j-1}^{l+1-j}(-a'b') \nonumber \\
    & \qquad \qquad \quad + \frac{b^k (b')^{-l-1}}{(k-1)! \, (l+1)} \sum_{j=l+1}^k  \left( \frac{b'}{b} \right)^{j}  L_{j-1}^{k-j}(-ab)  L_l^{j-l-1}(-a' b') \nonumber \\
    & \qquad \qquad \quad + \frac{a b^{k+1} (b')^{-l-1}}{k! \, (l+1)} \sum_{j=l+1}^k  \left( \frac{b'}{b} \right)^{j}  L_{j-1}^{k-j+1}(-ab) L_l^{j-l-1}(-a' b') \nonumber \\
    & \qquad \qquad \quad + \frac{a}{l+1} \left[ g_{kl} - \frac{(a')^l b^k}{k!} \sum_{j=0}^{k-1} j! \left( \frac{b'}{b} \right)^j L_j^{k-j} (-ab) \frac{(a'b')^{-\min(l,j)}}{\max(l,j)!} L_{\min(l,j)}^{|l-j|}(-a'b') \right]    \nonumber \\
    & \qquad \qquad \quad + \frac{1}{l+1} \left[ g_{k-1,l} - \frac{(a')^l b^{k-1}}{(k-1)!} \sum_{j=0}^{k-1} j! \left( \frac{b'}{b} \right)^j L_j^{k-j-1}(-ab) \frac{(a'b')^{-\min(l,j)}}{\max(l,j)!} L_{\min(l,j)}^{|l-j|} (-a'b') \right]   \nonumber \\
    & =: \frac{1}{l+1} g_{k-1,l} + \frac{a+a'}{l+1} g_{kl} + \frac{1}{k! \, (l+1)} R \, ,
\end{align}
where we have defined
\begin{align}
    R  := & \frac{(a')^{l+1} b^k}{l!} \sum_{j=1}^l j! (a'b)^{-j} L_{j}^{k-j} (-ab)  L_{j-1}^{l+1-j}(-a'b') + k b^k (b')^{-l-1} \sum_{j=l+1}^k  \left( \frac{b'}{b} \right)^{j}  L_{j-1}^{k-j}(-ab)  L_l^{j-l-1}(-a' b')  \nonumber \\
    &  + a b^{k+1} (b')^{-l-1} \sum_{j=l+1}^k  \left( \frac{b'}{b} \right)^{j}  L_{j-1}^{k-j+1}(-ab) L_l^{j-l-1}(-a' b') - a (a')^l b^k \sum_{j=0}^{l-1} j! \left( \frac{b'}{b} \right)^j L_j^{k-j} (-ab) \frac{(a'b')^{-j}}{l!} L_{j}^{l-j}(-a'b') \nonumber \\
    & - a (a')^l b^k \sum_{j=l}^{k-1} j! \left( \frac{b'}{b} \right)^j L_j^{k-j} (-ab) \frac{(a'b')^{-l}}{j!} L_{l}^{j-l}(-a'b') - k (a')^l b^{k-1} \sum_{j=0}^{l-1} j! \left( \frac{b'}{b} \right)^j L_j^{k-j-1}(-ab) \frac{(a'b')^{-j}}{l!} L_{j}^{l-j} (-a'b')\nonumber \\
    &  - k (a')^l b^{k-1} \sum_{j=l}^{k-1} j! \left( \frac{b'}{b} \right)^j L_j^{k-j-1}(-ab) \frac{(a'b')^{-l}}{j!} L_{l}^{j-l} (-a'b') \nonumber \\
    = &   \frac{(a')^{l+1} b^k}{l!} \sum_{j=1}^l j! (a'b)^{-j} L_{j}^{k-j} (-ab)  L_{j-1}^{l+1-j}(-a'b') + k b^{k-1} (b')^{-l} \sum_{j=l}^{k-1}  \left( \frac{b'}{b} \right)^{j}  L_{j}^{k-j-1}(-ab)  L_l^{j-l}(-a' b')\nonumber \\
    & + a b^{k} (b')^{-l} \sum_{j=l}^{k-1}  \left( \frac{b'}{b} \right)^{j}  L_{j}^{k-j}(-ab) L_l^{j-l}(-a' b') - \frac{a (a')^l b^k}{l!} \sum_{j=0}^{l-1} j!  (a'b)^{-j} L_j^{k-j} (-ab)  L_{j}^{l-j}(-a'b') \nonumber \\
    &  - a  b^k (b')^{-l} \sum_{j=l}^{k-1}  \left( \frac{b'}{b} \right)^j L_j^{k-j} (-ab)  L_{l}^{j-l}(-a'b') - \frac{k (a')^l b^{k-1}}{l!} \sum_{j=0}^{l-1} j!  (a'b)^{-j} L_j^{k-j-1}(-ab)  L_{j}^{l-j} (-a'b')\nonumber \\
    &  - k  b^{k-1} (b')^{-l} \sum_{j=l}^{k-1}  \left( \frac{b'}{b} \right)^j L_j^{k-j-1}(-ab)  L_{l}^{j-l} (-a'b') \, .
\end{align}
Here the second term cancels the last and the third cancels the fifth, so we have
\begin{align}
    R & = \frac{(a')^l b^{k-1}}{l!} \Bigg\{ a' b \sum_{j=1}^l j! (a'b)^{-j} L_{j}^{k-j} (-ab)  L_{j-1}^{l+1-j}(-a'b') \nonumber \\
    & \qquad \qquad \qquad - ab \sum_{j=0}^{l-1} j!  (a'b)^{-j}  L_j^{k-j} (-ab)  L_{j}^{l-j}(-a'b')  - k \sum_{j=0}^{l-1} j! (a'b)^{-j}  L_j^{k-j-1}(-ab)  L_{j}^{l-j} (-a'b') \Bigg\} \nonumber \\
    & = \frac{(a')^l b^{k-1}}{l!} \Bigg\{ a' b \sum_{j=1}^l j! (a'b)^{-j} L_{j}^{k-j} (-ab)  L_{j-1}^{l+1-j}(-a'b') -  \sum_{j=0}^{l-1} j!  (a'b)^{-j} (j+1) L_{j+1}^{k-j-1} (-ab)  L_{j}^{l-j}(-a'b') \Bigg\} \nonumber \\
    & = 0 \, ,
\end{align}
where in the second equality we have used \eqref{eq:nLn}.
\end{proof}

\section{Leggett-Garg CHSH inequality violation}
\label{app:lgviolation}
Here we compute the maximal violation of the Leggett-Garg CHSH inequality \eqref{eq:lgchsh} using states and measurements as obtained in the macroscopic limit (see main text). For convenience, we will work with phase-space representation of operators. In particular, we define
\begin{align}
    W_\rho(x,p) & := \frac{1}{\pi} \int_{-\infty}^{+\infty} ds \, e^{2 i ps} \mel{x-s}{\rho}{x+s} \, , \nonumber \\
    E(x,p) & := 2 \int_{-\infty}^{+\infty} ds \, e^{-2 ips} \ket{x-s} \bra{x+s} \, ,
\end{align}
so that $\rho = \int dx \, dp \, W_\rho(x,p) E(x,p)$. Defining $K_\j(\x) := K(\x | \cos \j \, \s_x + \sin \j \, \s_y)$ corresponding to the case where Alice measures spin in the $(\cos \j , \sin \j , 0)$ direction, we have
\begin{align}
    K_0(\x) \rho K_0(\x)^\dagger & = \int_{-\infty}^{+\infty} dx \, dp \, W_\rho(x,p) K_0(\x) E(x,p) K_0(\x)^\dagger \nonumber \\
    & = \int_{-\infty}^{+\infty} dx \, dp \, W_\rho(x,p) \cdot 2 \int_{-\infty}^{+\infty} ds \, e^{-2ip s} K_0(\x) \ket{x-s} \bra{x+s} K_0(\x)^\dagger \nonumber \\
    & = 2 \int_{-\infty}^{+\infty} dx \, dp \, ds \, W_\rho(x,p)   e^{-2ip s} \int_{-\infty}^{+\infty} dy \, \frac{e^{-(y-\x)^2/(2 \s^2)}}{(\pi \s^2)^{1/4}} \, \ket{y} \braket{y}{x-s} \int_{-\infty}^{+\infty} dy' \frac{e^{-(y'-\x)^2/(2 \s^2)}}{( \pi \s^2)^{1/4}}  \braket{x+s}{y'} \bra{y'} \nonumber \\
    & =\frac{2}{\sqrt{\pi \s^2}} \int_{-\infty}^{+\infty} dx \, dp \, ds \, W_\rho(x,p)   e^{-2ip s} \, e^{-(x-s-\x)^2/(2 \s^2)} \, e^{-(x+s-\x)^2/(2 \s^2)}  \, \ket{x-s} \bra{x+s} \nonumber \\
    & = \frac{2}{\sqrt{\pi \s^2}} \int_{-\infty}^{+\infty} dx \, dp \, ds \, W_\rho(x,p)   \, e^{-2ip s} \, e^{-(x-\x)^2/\s^2}   \, e^{-s^2/\s^2}  \,  \ket{x-s} \bra{x+s} \, .
\end{align}
Therefore
\begin{align}
    P_\j(\x,\h) & = \tr \Big[ K_\j(\h) K_0(\x) \rho K_0(\x)^\dagger K_\j(\h)^\dagger \Big] \nonumber \\
    & = \frac{2}{\sqrt{\pi \s^2}} \int_{-\infty}^{+\infty} dx \, dp \, ds \, W_\rho(x,p)   \, e^{-2ip s} \, e^{-(x-\x)^2/\s^2}   \, e^{-s^2/\s^2}  \, \tr \Big[  K_\j(\h)  \ket{x-s} \bra{x+s} K_\j(\h)^\dagger \Big] \nonumber \\
    & = \frac{2}{\sqrt{\pi \s^2}} \int_{-\infty}^{+\infty} dx \, dp \, ds \, W_\rho(x,p)   \, e^{-2ip s} \, e^{-(x-\x)^2/\s^2}   \, e^{-s^2/\s^2}  \, \int_{-\infty}^{+\infty} dy  \frac{e^{-(y-\h)^2/\s^2}}{\sqrt{\pi \s^2}} \braket{x+s}{y}_\j \, _\j\braket{y}{x-s} \, .
\end{align}
Now, using that \cite{cvmub}
\begin{align}
    \braket{x}{y}_\j = \frac{1}{\sqrt{2 \pi |\sin \j |}} \exp \left\{ - \frac{i \cos \j}{2 \sin \j} \left( x- \frac{y}{\cos \j} \right)^2 \right\} \, ,
\end{align}
the above expression reads
\begin{align}
    P_\j (\x,\h) & = \frac{1}{\pi^2  \s^2 | \sin \j |} \int_{-\infty}^{+\infty} dx \, dp \, ds \, dy \, W_\rho(x,p) \, e^{\a(x,p,s,y)} \, ,
\end{align}
where
\begin{align}
    \a(x,p,s,y) & =  -2ip s - \frac{(x-\x)^2}{\s^2} - \frac{ s^2}{\s^2} - \frac{(y-\h)^2}{\s^2}   -\frac{i \cos \j}{2 \sin \j} \left( x + s - \frac{y}{\cos \j} \right)^2  + \frac{i \cos \j }{2 \sin \j} \left( x - s - \frac{y}{\cos \j} \right)^2  \nonumber \\
    & =   -2ip s - \frac{(x-\x)^2}{\s^2} - \frac{ s^2}{\s^2} - \frac{(y-\h)^2}{\s^2} - \frac{2 i \cos \j}{\sin \j} \left( x - \frac{y}{\cos \j} \right) s  \, .
\end{align}
Integrating in $s$ gives
\begin{align}
    P_\j(\x,\h) & = \frac{1}{\sqrt{\pi^3 \s^2} |\sin \j |} \int_{-\infty}^{+\infty} dx \, dp \, dy \, W_\rho(x,p) \, e^{\b(x,p,y)} \, ,
\end{align}
where 
\begin{align}
    \b(x,p,y) & = - \frac{(x-\x)^2}{\s^2}  - \frac{(y-\h)^2}{\s^2} - \frac{\s^2}{4} \left[ 2 p + \frac{2 \cos \j}{\sin \j} \left( x - \frac{y}{\cos \j} \right) \right]^2   \nonumber \\
    & = - \frac{(x-\x)^2}{\s^2}  - \frac{y^2}{\s^2} - \frac{\h^2}{\s^2} + \frac{2 \h y}{\s^2} - \s^2 p^2  - \s^2 \frac{\cos^2 \j}{\sin^2 \j} \left( x - \frac{y}{\cos \j} \right)^2 - 2 \s^2 p \frac{\cos \j}{\sin \j} \left( x - \frac{y}{\cos \j} \right) \nonumber \\
    & = - \frac{(x-\x)^2}{\s^2}  - \frac{y^2}{\s^2}  - \frac{\h^2}{\s^2} + \frac{2 \h y}{\s^2} - \s^2 p^2  -  \frac{\s^2 x^2 \cos^2 \j}{\sin^2 \j}  - \frac{\s^2 y^2}{\sin^2 \j} + \frac{2 \s^2 x y \cos \j}{\sin^2 \j} - \frac{2 \s^2 p x \cos \j}{\sin \j}  +   \frac{2 \s^2 p y }{\sin \j} \nonumber \\
    & = - \frac{(x-\x)^2}{\s^2}  - \frac{\h^2}{\s^2}  - \frac{\s^2}{\sin^2 \j} p^2 \sin^2 \j  -  \frac{\s^2}{\sin^2 \j} x^2 \cos^2 \j   - \frac{2 \s^2}{\sin^2 \j}   x \cos \j \, p \sin \j \nonumber \\
    & \quad \,- \frac{ \s^4 + \sin^2 \j}{ \s^2 \sin^2 \j} y^2 + \left( \frac{2\h}{\s^2} + \frac{2 \s^2 x \cos \j}{\sin^2 \j} + \frac{2 \s^2 p}{\sin \j} \right) y\, .
\end{align}
Integrating in $y$ gives
\begin{align}
    P_\j(\x,\h) & = \frac{1}{ \pi \sqrt{\s^4 + \sin^2 \j}}  \int_{-\infty}^{+\infty} dx \, dp \, W_\rho(x,p) \, e^{\g(x,p)} \, ,
\end{align}
where
\begin{align}
    \g(x,p) & =  - \frac{(x-\x)^2}{\s^2}  - \frac{\h^2}{\s^2}  - \frac{\s^2}{\sin^2 \j} (x \cos \j + p \sin \j)^2 + \frac{1}{\s^2} \frac{\sin^2 \j}{\s^4 + \sin^2 \j} \left(\h + \frac{\s^4}{\sin^2 \j} (x \cos \j + p \sin \j) \right)^2 \nonumber \\
    & =  - \frac{(x-\x)^2}{\s^2}  - \frac{1}{\s^2} \left( 1 - \frac{\sin^2 \j}{\s^4 + \sin^2 \j} \right) \h^2 - \frac{\s^2}{\sin^2 \j} \left( 1- \frac{\s^4}{\s^4 + \sin^2 \j} \right) (x \cos \j + p \sin \j)^2 \nonumber \\
    & \quad \, + \frac{2 \s^2}{\s^4 + \sin^2 \j} \h (x \cos \j + p \sin \j ) \nonumber \\
    & = - \frac{(x-\x)^2}{\s^2}  - \frac{\s^2}{\s^4 + \sin^2 \j} \h^2 - \frac{\s^2}{\s^4 + \sin^2 \j} (x \cos \j + p \sin \j)^2 + \frac{2 \s^2}{\s^4 + \sin^2 \j} \h (x \cos \j + p \sin \j ) \nonumber \\
    & = - \frac{(x-\x)^2}{\s^2} - \frac{\s^2}{\s^4 + \sin^2 \j} \big( x \cos \j + p \sin \j - \h \big)^2 \, .
\end{align}
So we have
\begin{align}
    P_\j (\x,\h) = \frac{1}{\pi \sqrt{\s^4 + \sin^2 \j}} \int_{-\infty}^{+\infty} dx \, dp \, W_\rho (x, p ) \, \exp \left\{ - \frac{(x-\x)^2}{\s^2} - \frac{\s^2}{\s^4 + \sin^2 \j} \big( x \cos \j + p \sin \j - \h \big)^2 \right\} \, .
\end{align}
Then, the correlator of the random variables $\sgn(\x)$ and $\sgn(\h)$ is
\begin{align}
    \avg{\sgn(\x) \sgn(\h)}_\j & = \int_{-\infty}^{+\infty} d\x \, d\h \, \sgn(\x) \sgn(\h) P_\j(\x, \h) \nonumber \\
    & = \frac{1}{\pi \sqrt{\s^4 + \sin^2 \j}} \int_{-\infty}^{+\infty} dx \, dp \, W_\rho (x, p ) \, \int_{-\infty}^{+\infty} d\x \, \sgn(\x) e^{ - \frac{(x-\x)^2}{\s^2}} \int_{-\infty}^{+\infty} d\h \, \sgn(\h) e^{ - \frac{\s^2 ( \h - x \cos \j - p \sin \j  )^2}{\s^4 + \sin^2 \j}  } \nonumber \\
    & = \frac{1}{\pi \sqrt{\s^4 + \sin^2 \j}} \int_{-\infty}^{+\infty} dx \, dp \, W_\rho (x, p ) \, \sqrt{\pi \s^2} \, \erf \left( \frac{x}{\s} \right) \sqrt{ \pi \frac{\s^4 + \sin^2 \j}{\s^2}} \, \erf \left( \frac{x \cos \j + p \sin \j}{\sqrt{\frac{\s^4 + \sin^2 \j}{\s^2}}} \right) \nonumber \\
    & = \int_{-\infty}^{+\infty} dx \, dp \, W_\rho (x, p ) \, \erf \left( \frac{x}{\s} \right)  \erf \left( \frac{x \cos \j + p \sin \j}{\sqrt{\frac{\s^4 + \sin^2 \j}{\s^2}}} \right) \, .
\label{eq:correlator}
\end{align}
Now let us compute the Wigner function of the state $\rho = \sum_{k,l} c_{kl} \ket{l} \bra{k}$:
\begin{align}
    W_\rho (x,p) & = \frac{1}{\pi} \sum_{k,l}  c_{kl} \int_{-\infty}^{+\infty} ds \, e^{2ips} \braket{x-s}{l} \braket{k}{x+s} \nonumber \\
    & =  \frac{1}{\pi\sqrt{\pi}} \sum_{k,l} \frac{ c_{kl}}{\sqrt{2^k \, k! \, 2^l \, l!}} \int_{-\infty}^{+\infty} ds \, e^{2ips} e^{-(x-s)^2/2} H_l(x-s) e^{- (x+s)^2/2} H_k (x+s) \nonumber \\
    & =  \frac{e^{- x^2}}{\pi\sqrt{\pi}} \sum_{k,l} \frac{ c_{kl}}{\sqrt{2^k \, k! \, 2^l \, l!}} \int_{-\infty}^{+\infty} ds \, e^{-s^2 +  2ips} \, H_l (x-s)  H_k(x+s)   \nonumber \\
    & = \frac{e^{-x^2}}{\pi\sqrt{\pi}} \sum_{k,l} \frac{c_{kl} }{\sqrt{2^k \, k! \, 2^l \, l!}} \int_{-\infty}^{+\infty} ds \, e^{-(s - ip)^2 - p^2} \,    H_l(x - s ) \, H_k( x+s)  \, .
\end{align}
Defining $\z := s - ip$ we have
\begin{align}
    W_\rho (x,p) & =  \frac{e^{- x^2 - p^2}}{\pi\sqrt{\pi}} \sum_{k,l} \frac{ c_{kl}}{\sqrt{2^k \, k! \, 2^l \, l!}} \int_{-\infty-ip}^{+\infty-ip} d\z \, e^{-\z^2} \,    H_l \big( x-\z - ip\big) \, H_k \big( x+\z +ip \big) \, .
\end{align}
Since the integrand is holomorphic, we can set to zero the imaginary shift in the integration contour. Then, using the property \cite{gradshteyn}
\begin{align}
    \int_{-\infty}^{+\infty} dx \, e^{-x^2} \,    H_k (x + y) \, H_l ( x +z) = 2^{\max(k,l)} \sqrt{\pi} \, \min(k,l)! \,  y^{k-\min(k,l)} z^{l-\min(k,l)} L_{\min(k,l)}^{|k-l|} (-2yz) \, ,
\end{align}
we can write
\begin{align}
    W_\rho (x,p) & =  \frac{ e^{- x^2 - p^2}}{\pi} \sum_{k,l}  c_{kl} \, (-1)^{l} \sqrt{\frac{2^{\max(k,l)} \min(k,l)!}{2^{\min(k,l)} \max(k,l)!}}  \nonumber \\
    & \qquad \qquad \qquad \qquad \qquad  \cdot   \big( -x+ip \big)^{l-\min(k,l)} \big(x + ip)^{k-\min(k,l)} L_{\min(k,l)}^{|k-l|} \big(-2(-  x+ip)(x+ip) \big) \nonumber \\
    & = \frac{ e^{-x^2 - p^2}}{\pi} \Bigg\{  \sum_{k\geq l}  c_{kl} (-1)^l \sqrt{\frac{2^k \, l!}{2^l \, k!}} (x+ip)^{k-l} \, L_l^{k-l} \big( 2 x^2 + 2 p^2 \big) \nonumber \\
    & \qquad \qquad \qquad \qquad \qquad +\sum_{k< l} c_{kl} (-1)^l \sqrt{\frac{2^l \, k!}{2^k \, l!}} (- x+ip)^{l-k} \, L_k^{l-k} \big( 2 x^2 + 2 p^2 \big) \Bigg\} \, .
\end{align}
Plugging this back into the correlator \eqref{eq:correlator} yields
\begin{align}
    \langle \sgn (\x) \sgn (\h) \rangle_\j = & \frac{1}{\pi}   \sum_{k\geq l}  c_{kl} (-1)^l \sqrt{\frac{2^k \, l!}{2^l \, k!}} \int_{- \infty}^{+ \infty} dx dp \,  e^{-x^2 - p^2} (x+ip)^{k-l} \nonumber \\
    & \qquad \qquad \qquad \qquad \qquad \qquad \qquad \quad  \cdot L_l^{k-l} \big( 2 x^2 + 2 p^2 \big)  \erf \left( \frac{x}{\s} \right)  \erf \left( \frac{x \cos \j + p \sin \j}{\sqrt{\frac{\s^4 + \sin^2 \j}{\s^2}}} \right) \nonumber \\
    &  + \frac{1}{\pi}  \sum_{k< l} c_{kl} (-1)^l \sqrt{\frac{2^l \, k!}{2^k \, l!}} \int_{- \infty}^{+ \infty} dx dp \,  e^{-x^2 - p^2} (- x+ip)^{l-k}  \nonumber \\
    & \qquad \qquad \qquad \qquad \qquad \qquad \qquad \quad  \cdot L_k^{l-k} \big( 2 x^2 + 2 p^2 \big) \erf \left( \frac{x}{\s} \right)  \erf \left( \frac{x \cos \j + p \sin \j}{\sqrt{\frac{\s^4 + \sin^2 \j}{\s^2}}} \right) \, .
\end{align}
Let us restrict to the 3-dimensional subspace spanned by the first three excitations in the Hilbert space so that the sums over $k$ and $l$ only run until 2. In this case we can analytically perform the integrals using the identities \cite{korotkov}
\begin{align}
    \int_{-\infty}^{+\infty} dx \, e^{-x^2 } \, \erf (ax+b) & = \sqrt{\pi} \, \erf \left( \frac{b}{\sqrt{1+a^2}} \right) \, , \nonumber \\
    \int_{-\infty}^{+\infty} dx \, x \,e^{-x^2 } \, \erf (ax+b) & = \frac{a}{\sqrt{1+a^2}} \exp \left( -\frac{b^2}{1+a^2} \right) \, , \nonumber \\
    \int_{-\infty}^{+\infty} dx \, x^{2m} e^{-a x^2} \, \erf (b x) \erf (c x) & = 2 \frac{(-1)^m}{\sqrt{\pi}} \frac{\partial^m}{\partial a^m} \left[ \frac{1}{\sqrt{a}} \arctan \left( \frac{b c}{\sqrt{a^2 +a (b^2 +c^2)}} \right) \right]  \, , \quad (a > 0 \, , \, \,  m=0, 1, \dots ) \, ,\nonumber \\
    \int_{-\infty}^{+\infty} dx \, x^2 \, e^{-x^2} \, \erf (a x+b) & = \frac{\sqrt{\pi}}{2} \erf \left( \frac{b}{\sqrt{1+a^2}} \right) - \frac{a^2 b}{(1+a^2)^{3/2}} \exp \left( - \frac{b^2}{1+a^2} \right) \, ,\nonumber \\
    \int_{-\infty}^{+\infty} dx \, x^3 \, e^{-x^2} \, \erf (a x) & = \frac{a}{\sqrt{1+a^2}} \left( 1 + \frac{1}{2} \frac{1}{1+a^2} \right)  \, ,\nonumber \\
     \int_{-\infty}^{+\infty} dx \, x^4 \, e^{-x^2} \, \erf (a x +b) & = \frac{3 \sqrt{\pi}}{4} \erf \left( \frac{b}{\sqrt{1+a^2}} \right) - \frac{3}{2} \frac{a^2 b}{(1+a^2)^{3/2}} \left[1 + \frac{2 a^2 b^2}{3(1+a^2)^2} + \frac{1}{1+a^2 }\right]  \exp \left( - \frac{b^2}{1+a^2} \right) \, .
\end{align}
The results are compact expressions, but too long to show here. Finally, the Leggett-Garg CHSH parameter is
\begin{align}
    C = \langle \sgn (\x) \sgn (\h) \rangle_{\j_A^{(1)}-\j_B^{(1)}} + \langle \sgn (\x) \sgn (\h) \rangle_{\j_A^{(1)}-\j_B^{(2)}} +  \langle \sgn (\x) \sgn (\h) \rangle_{\j_A^{(2)}-\j_B^{(1)}} -\langle \sgn (\x) \sgn (\h) \rangle_{\j_A^{(2)}-\j_B^{(2)}} \, .
\end{align}
Maximizing $C$ over the four angles and the coefficients $c_{kl}$ we have the following violation of the inequality $C\geq 2$:
\begin{align} 
    C = \frac{2}{675 \pi} \Big( 577 + \sqrt{1244179} + 2700 \arctan (1/3) \Big) \simeq 2.416 \, ,
\end{align}
obtained for the angles 
\begin{align}
    \j_A^{(1)} = \frac{\pi}{4} \, , \quad \j_A^{(2)} = \frac{3 \pi}{4} \, , \quad \j_B^{(1)} = \frac{\pi}{2} \, , \quad \j_B^{(2)} = 0 \, ,  
\end{align}
and the state 
\begin{align}
    \ket{\psi} = \sqrt{\frac{1}{2} - \frac{577}{2 \sqrt{1244179}}} \, \ket{0} + \sqrt{\frac{1}{2} + \frac{577}{2 \sqrt{1244179}}} \,  \ket{2} \, .
\end{align}

\end{document}